\documentclass{article}
\usepackage[a4paper,margin=1in]{geometry}

\usepackage[utf8]{inputenc} 
\usepackage[T1]{fontenc}    
\usepackage{hyperref}       
\usepackage{url}            
\usepackage{booktabs}       
\usepackage{amsfonts}       
\usepackage{nicefrac}       
\usepackage{microtype}      
\usepackage{xcolor}         
\usepackage{soul}
\usepackage{comment}
\usepackage{bbm}
\usepackage{amssymb,amsmath,mathrsfs,amsthm}
\usepackage{comment,multirow,graphicx,bbm,subcaption,caption,multirow,mathtools}
\usepackage{mathabx,algorithmic,algorithm}
\usepackage{xcolor}
\usepackage{bm}
\usepackage{upgreek}
\newtheorem{theorem}{Theorem}
\newtheorem{lemma}{Lemma}
\newtheorem{proposition}{Proposition}
\newtheorem{corollary}{Corollary}
\newtheorem{assumption}{Assumption}
\newtheorem{definition}{Definition}

\newenvironment{hproof}{%
  \proof}{\endproof}

\title{Heterogeneous Beliefs and Multi-Population Learning in Network Games}
\usepackage{authblk}
\author[1]{Shuyue Hu}
\author[2]{Harold Soh}
\author[3]{Georgios Piliouras}
\affil[1]{Shanghai Artificial Intelligence Laboratory}
\affil[2]{National University of Singapore}
\affil[3]{Singapore University of Technology and Design}

\date{}
\begin{document}

\maketitle
\begin{abstract}
The effect of population heterogeneity in multi-agent learning is practically relevant but remains far from being well-understood. 
Motivated by this, we introduce a model of multi-population learning that allows for heterogeneous beliefs within each population and where agents respond to their beliefs via smooth fictitious play (SFP).
We show that the system state --- a probability distribution over beliefs --- evolves according to a system of partial differential equations.
We establish the convergence of SFP to Quantal Response Equilibria in different classes of games capturing both network competition as well as network coordination. 
We also prove that the beliefs will eventually homogenize in all network games.
Although the initial belief heterogeneity disappears in the limit, we show that it plays a crucial role for equilibrium selection in the case of coordination games as it helps select highly desirable equilibria. Contrary, in the case of network competition, the resulting limit behavior is independent of the initialization of beliefs, even when the underlying game has many distinct Nash equilibria.
\end{abstract}

\section{Introduction}

Smooth Fictitious play (SFP) and variants thereof are arguably amongst the most well-studied learning models in AI and game theory
\cite{benaim2013consistency,benaim1999mixed,hofbauer2007evolution,hopkins1999note,fudenberg1993learning,hofbauer2005learning,perolat2018actor,perrin2020fictitious,xie2020provable,hernandez2014selective,han2020deep}. 
SFP describes a belief-based learning process: agents form beliefs about the play of opponents and update their beliefs based on observations.
Informally, an agent's belief can be thought as reflecting how likely its opponents will play each strategy.
During game plays, each agent plays smoothed best responses to its beliefs. 
Much of the literature of SFP is framed in the context of \emph{homogeneous beliefs} models where all agents in a given role have the same beliefs. 
This includes models with one agent in each player role \cite{benaim1999mixed,benaim2013consistency,swenson2019smooth} as well as models with a single population but in which all agents have the same beliefs \cite{hofbauer2007evolution,hopkins1999note}.
 SFP are known to converge in large classes of homogeneous beliefs models (e.g., most 2-player games \cite{fudenberg1993learning,hofbauer2005learning,benaim1999mixed}).  However, in the context of \emph{heterogeneous beliefs}, where agents in a population have different beliefs, SFP has been explored to a less extent.

The study of heterogeneous beliefs (or more broadly speaking, population heterogeneity) is important and practically relevant.
From multi-agent system perspective, heterogeneous beliefs widely exist in many applications, such as traffic management, online trading and video game playing. 
For example, it is natural to expect that public opinions generally diverge on autonomous vehicles and that people have different beliefs about the behaviors of taxi drivers vs non-professional drivers.
From machine learning perspective, recent empirical advances hint that injecting heterogeneity potentially accelerates population-based training of neural networks and improves learning performance \cite{jaderberg2017population,leung2021self,zhao2021maximum}. From game theory perspective, considering heterogeneity of beliefs  better explains results of some human experiments \cite{fudenberg1993steady,fudenberg1997measuring}.

Heterogeneous beliefs models of SFP are not entirely new. In the pioneering work \cite{fudenberg2011heterogeneous}, Fudenberg and Takahashi  examine the heterogeneity issue in 2-population settings 
by appealing to techniques from the \emph{stochastic approximation} theory. This approach, which is typical in the SFP literature, relates the limit behavior of each individual to an ordinary differential equation (ODE) and has yielded significant insights for many homogeneous beliefs models \cite{benaim1999mixed,benaim2013consistency,hofbauer2005learning,swenson2019smooth}. However, this approach, as also noted by Fudenberg and Takahashi,  ``does not provide very precise estimates of the effect of the initial
condition of the system.''
Consider an example of a population of agents each can choose between two pure strategies $s_1$ and $s_2$. Let us imagine two cases: (i) every agents in the population share the same belief that their opponents play a mixed strategy choosing $s_1$ and $s_2$ with equal probability $0.5$, and (ii) half of the agents believe that their opponents determinedly play the pure strategy $s_1$ and the other half believe that their opponents determinedly play the pure strategy $s_2$. The stochastic approximation approach would generally treat these two cases equally, providing little information about the heterogeneity in beliefs as well as its consequential effects on the system evolution.
This drives our motivating questions:

 \emph{How does heterogeneous populations 
 evolve under SFP? How much and under what conditions does the heterogeneity in beliefs affect their long-term behaviors?}
 
\textbf{Model and Solutions.} In this paper, we study the dynamics of SFP in general classes of multi-population network games that allow for heterogeneous beliefs.
In a multi-population network game, each vertex of the network represents a population (continuum) of agents, and each edge represents a series of 2-player subgames between two neighboring populations.
Note that multi-population network games include all the 2-population games considered in \cite{fudenberg2011heterogeneous} and are representation of subclasses of real-world systems where the graph structure is evident \cite{czechowski2021poincar}.
We consider that for a certain population, individual agents form separate beliefs about each neighbor population and observe the mean strategy play of that population.
Taking a approach different from stochastic approximation, we define the \emph{system state} as a probability measure over the space of beliefs, which allows us to precisely examine the impact of heterogeneous beliefs on system evolution. This probability measure changes over time in response to agents' learning. Thus, the main challenge is to analyze the evolution of the measure, which in general requires the development of new techniques.

As a starting point, we establish a system of partial differential equations (PDEs) to track the evolution of the measure in continuous time limit (Proposition \ref{theorem:pde}). 
The PDEs that we derive are akin to the \emph{continuity equations}\footnote{The continuity equation is a PDE that describes the transport phenomena of some quantity
(e.g., mass, energy, momentum and other conserved quantities) in a physical system. } commonly encountered in physics and do not allow for a general solution. 
Appealing to \emph{moment closure approximation} \cite{gillespie2009moment}, we circumvent the need of solving the PDEs and directly analyze the dynamics of the mean and variance (Proposition \ref{prop:meanbelief} and Theorem \ref{theorem:variance}). 
As one of our key results, we prove that the variance of beliefs always decays quadratically fast with time in \emph{all} network games (Theorem \ref{theorem:variance}).
Put differently, eventually, beliefs will homogenize and the distribution of beliefs will collapse to a single point,
regardless of initial distributions of beliefs, 2-player subgames that agents play, and the number of populations and strategies.
This result is non-trivial and perhaps somewhat counterintuitive. Afterall, one may find it more natural to expect that the distribution of beliefs would converge to some distribution rather than a single point, as evidenced by recent studies on Q-learning and Cross learning \cite{hu2019modelling,hu2022,lahkar2013reinforcement}.

Technically, the eventual belief homogenization  has a significant implication --- it informally hints that the asymptotic system state of initially heterogeneous systems are likely to be the same as in homogeneous systems.
We show that the fixed point of SFP correspond to Quantual Response Equilibria (QRE)\footnote{QRE is a game theoretic solution concept under bounded rationality. By QRE, in this paper we refer to their canonical form also referred to as logit equilibria or logit QRE in the literature~\cite{goeree2016quantal}.} in network games for both homogeneous and initially heterogeneous systems (Theorem \ref{prop:fixedpoint}).
As our main result, we establish the convergence of SFP to QRE in different classes of games capturing both network competition as well as network coordination, independent of belief initialization. 
Specifically, for competitive network games, we first prove via a Lyapunov argument that the SFP  converges to a unique QRE in homogeneous systems, even when the underlying  game has many distinct Nash equilibria (Theorem \ref{theorem:convergeHomo}). 
Then, we show that this convergence result can be carried over to initially heterogeneous systems (Theorem \ref{theorem:convergeHetero}), by leveraging that the mean belief dynamics of initially heterogeneous systems is \emph{asymptotically autonomous} \cite{markus1956asymptotically} with its limit dynamics being the belief dynamics of a homogeneous system (Lemma \ref{le:limit}).  
For coordination network games, we also prove the convergence to QRE for homogeneous and initially heterogeneous systems, in which the underlying network has star structure (Theorem \ref{theorem:convergePotential}). 

On the other hand, the eventual belief homogenization may lead to a \emph{misconception} that belief heterogeneity has little effect on system evolution.
Using an example of 2-population stag hunt games, we show that belief heterogeneity actually plays a crucial role in equilibrium selection, even though it eventually vanishes.
As shown in Figure 1, changing the variance of initial beliefs results in different limit behaviors, even when the mean of initial beliefs remains unchanged;
in particular, while a small variance leads to the less desirable equilibrium $(H,H)$,  a large variance leads to the payoff dominant equilibrium $(S,S)$.
Thus, in the case of network coordination, initial belief heterogeneity can help select the highly desirable equilibrium and provides 
interesting insights to the seminal thorny problem
of equilibrium selection~\cite{kirman1993frontiers}.
On the contrary, in the case of network competition, we prove (Theorems \ref{theorem:convergeHomo} and \ref{theorem:convergeHetero} on the convergence to a unique QRE in competitive network games) as well as showcase experimentally that the resulting limit behavior is independent of initialization of beliefs, even if the underlying game has many distinct Nash equilibria.

\begin{figure}[tb!]
    \centering
    \includegraphics[width=0.9\textwidth]{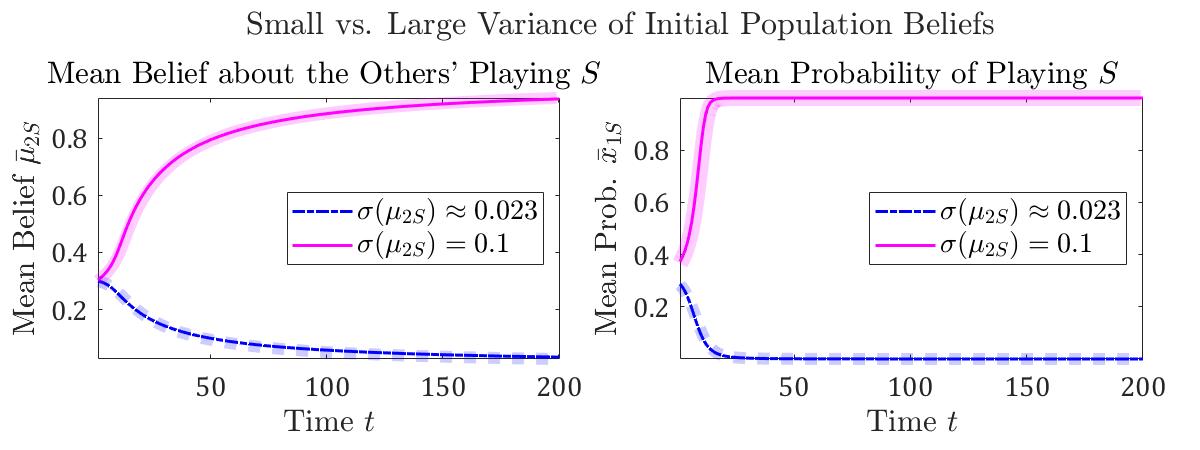}
    \caption{The system dynamics under the effects of different  variances of initial beliefs (thin lines: predictions of our PDE model, shaded wide lines: simulation results). 
    $\bar{\mu}_{2S}$ represents the mean belief about population $2$ and $\bar{x}_{1S}$ represents the mean probability of playing strategy $S$ in population $1$. Initially, we set the mean beliefs $\bar{\mu}_{2S}=\bar{\mu}_{1S} = 0.3$ (details of the setup are summarized in the supplementary).  
    Given the same initial mean belief, different initial variances $\sigma^2(\mu_{2S})$ lead to the convergence to different beliefs (the left panel) and even to different strategy choices (the right panel). In particular, a large initial variance helps select the payoff dominant equilibrium $(S,S)$ in stag hunt games.}
    \label{fig:SHgame}
\end{figure}

\paragraph{Related Works.}
SFP and its variants have recently attracted a lot of attention in AI research \cite{perolat2018actor,perrin2020fictitious,xie2020provable,hernandez2014selective,han2020deep}. 
There is a significant literature that analyze SFP in different models \cite{benaim1999mixed,ewerhart2020fictitious,hofbauer2007evolution,hofbauer2005learning}, and the paper that is most closely related to our work is \cite{fudenberg2011heterogeneous}. Fudenberg and Takahashi \cite{fudenberg2011heterogeneous} also examines the heterogeneity issue and anticipate belief homogenization in the limit under 2-population settings. In this paper, we consider multi-population network games, which is a generalization of their setting.\footnote{The analysis presented in this paper covers all generic 2-population network games, all generic bipartite network games where the game played on each edge is the same along all edges, and all weighted zero-sum games which do not require the graph to be bipartite nor to have the same game played on each edge. } Moreover, our approach is more fundamental, as the PDEs that we derive can provide much richer information about the system evolution and thus precisely estimates the temporal effects of heterogeneity, which is generally intractable in \cite{fudenberg2011heterogeneous}.
Therefore, using our approach, we are able to show an interesting finding — the initial heterogeneity plays a crucial role in equilibrium selection (Figure 1) — which unfortunately cannot be shown using the approach in \cite{fudenberg2011heterogeneous}. Last but not least, to our knowledge, our paper is the first work that presents a systematic study of smooth fictitious play in general classes of network games. 

On the other hand, networked multi-agent learning constitutes one of the current frontiers in AI and ML research \cite{zhang2018fully,liu2021local,gupta2020networked}. 
Recent theoretical advances on network games
provide conditions for learning behaviors to be not chaotic \cite{czechowski2022poincar,nagarajan2020chaos}, and investigate the convergence of Q-learning and continuous-time FP in the case of network competitions \cite{ewerhart2020fictitious,leonardos2021exploration}. However, \cite{ewerhart2020fictitious,leonardos2021exploration} consider that there is only one agent on each vertex, and hence their models are essentially for homogeneous systems.

Lahkar and Seymour \cite{lahkar2013reinforcement} and Hu et al. \cite{hu2019modelling,hu2022} also use the continuity equations as a tool to study population heterogeneity in multi-agent systems where a single population of agents applies Cross learning or Q-learning to play symmetric games. 
They either prove or numerically showcase that heterogeneity generally persists.
Our results complement these advances by showing that heterogeneity vanishes under SFP and that heterogeneity helps select highly desirable equilibria.
Moreover, methodologically, we establish new proof techniques for the convergence of learning dynamics in heterogeneous systems by leveraging seminal results (Lemmas \ref{le:1} and \ref{le:2}) from the asymptotically autonomous dynamical system literature, which may be of independent interest. 




\section{Preliminaries}

\paragraph{Population Network Games.}
A population network game (PNG)
$\Gamma=( N, ( V, E ),  (S_i,\omega_i)_{\forall i \in V}, (\mathbf{A}_{ij})_{(i,j)\in E}) $
consists of a multi-agent system $N$ distributed over a graph $(V, E)$, where $V = \{1, ..., n\}$ is the set of vertices each represents a population (continuum) of agents, and $E$ is the set of pairs, $(i,j)$, of population $i \neq j \in V$. 
For each population $i \in V$,  agents of this population has a finite set $S_i$ of pure strategies (or actions) with generic elements $s_{i} \in S_i$.
Agents may also use mixed strategies (or choice distributions).  For an arbitrary agent $k$ in population $i$, its mixed strategy is a vector $\mathbf{x}_{i}(k) \in \Delta_{i}$, where $\Delta_{i}$ is the simplex in $\mathbb{R}^{|S_i|}$ such that $\sum_{s_i \in S_i} x_{is_i}(k) =1 $ and $ x_{is_i}(k) \geq 0, \forall s_i \in S_i$. 
Each edge $(i,j) \in E$ defines a series of two-player subgames between populations $i$ and $j$, such that for a given time step, each agent in population $i$ is randomly paired up with another agent in population $j$ to play a two-player subgame. We denote the payoff matrices for agents of population $i$ and $j$ in these two-player subgames by $\mathbf{A}_{ij} \in \mathbb{R}^{|S_i|\times |S_j|}$ and 
$\mathbf{A}_{ji} \in \mathbb{R}^{|S_j|\times |S_i|}$, respectively.
Note that at a given time step, each agent chooses a (mixed or pure)  strategy and plays that strategy in all two-player subgames.
Let $\mathbf{x} = (\mathbf{x}_i, \{\mathbf{x}_{j}\}_{(i,j) \in E})$ be 
a mixed strategy profile,  where $\mathbf{x}_{i}$ (or $\mathbf{x}_{j}$) denotes a generic mixed strategy in population $i$ (or $j$). Given the mixed strategy profile $\mathbf{x}$, the expected payoff of using $\mathbf{x}_i$ in the game $\Gamma$ is
\begin{equation}
 r_i(\mathbf{x}) =   r_i(\mathbf{x}_i, \{\mathbf{x}_{j}\}_{(i,j) \in E}) \coloneqq \sum_{{(i,j)\in E}} \mathbf{x}_i^\top \mathbf{A}_{ij} \mathbf{x}_j.
\end{equation}
The game $\Gamma$ is \emph{competitive} (or \emph{weighted zero-sum}), if there exist positive constants $\omega_1, \ldots, \omega_n$ such that
\begin{equation}
    \sum_{i \in V} \omega_i r_i(\mathbf{x}) = \sum_{(i,j)\in E} \left( \omega_i \mathbf{x}_i^\top \mathbf{A}_{ij} \mathbf{x}_j + \omega_j \mathbf{x}_j^\top \mathbf{A}_{ji} \mathbf{x}_i \right ) = 0, \quad \forall \mathbf{x} \in \prod_{ i \in V} \Delta_i.
    \label{eq:zerosum}
\end{equation}
On the other hand, $\Gamma$ is a \emph{coordination} network game, if for each edge $(i,j)\in E$, the payoff matrices of the two-player subgame satisfy $\mathbf{A}_{ij} = \mathbf{A}_{ji}^\top.$ 

\paragraph{Smooth Fictitious Play.} 
SFP is a belief-based model for learning in games. In  SFP, agents 
form  beliefs about the play of opponents and respond to the beliefs via smooth best responses.
Given a game $\Gamma$, consider an arbitrary agent $k$ in a population $i\in V$.
Let $V_{i}=\{j\in V: (i,j)\in E\}$ be the set of neighbor populations. 
Agent $k$ maintains a weight  $\kappa^i_{js_j}(k)$ for each  opponent strategy $s_j \in S_j$ of a neighbor population $j\in V_i$.
Based on the weights, agent $k$ forms a belief about the neighbor population $j$, such that each opponent strategy $s_j$ is played with probability
\begin{equation}
    \mu_{js_j}^{i}(k) = \frac{\kappa^i_{js_j}(k)}{\sum_{s_j'\in S_j}\kappa^i_{js_j'}(k)}.
\end{equation}
Let $\bm{\upmu}_j^i(k)$ be the vector of beliefs with the $s_j$-\text{th} element equals  $\mu_{js_j}^i(k)$.
Agent $k$ forms separate beliefs for each neighbor population, and plays a smooth best response to the set of beliefs $\{\bm{\upmu}_j^i(k)\}_{j\in V_i}$. Given a game $\Gamma$, agent $k$'s expected payoff for using a pure strategy $s_i \in S_i$ is 
\begin{equation}
\quad u_{is_i}(k) = r_i(\mathbf{e}_{s_i}, \{\bm\upmu_j^i(k,t)\}_{j\in V_i})  = \sum_{j\in V_i} \mathbf{e}_{s_i}^\top \mathbf{A}_{ij} \bm{\upmu}_j^i(k)
\end{equation}
where $\mathbf{e}_{s_i}$ is a unit vector where the $s_i$-\text{th} element is $1$. The probability of playing strategy $s_i$ is then given by
\begin{equation}
\label{eq:exploration}
       x_{is_i}(k) = \frac{\exp({\beta u_{is_i}(k) )}}{\sum_{s_i' \in S_i} \exp({\beta u_{is_i'}(k)})} 
\end{equation}
where $\beta$ is a temperature (or the degree of rationality).
We consider that agents observe the mean mixed strategy of each neighbor population.
As such, at a given time step $t$, agent $k$ updates the weights for each opponent strategy $s_j \in S_j, j\in V_i$ as follows:
\begin{equation}
    \kappa_{js_j}^i(k,t+1) = \kappa_{js_j}^{i}(k,t) + \bar{x}_{js_j}(t)
    \label{eq:weight}
\end{equation}
where $\bar{x}_{js_j}$ is the mean probability of  playing  strategy $s_j$ in  population $j$, i.e., $\bar{x}_{js_j} = \frac{1}{n_j} \sum_{l \in \text{population }j}x_{js_j}(l)$  with the number of agents denoted by $n_j$. For simplicity, we assume the initial sum of weights $\sum_{s_j \in S_j} {\kappa_{js_j}^{i}(k,0)}$ to be the same for every agent in the system $N$ and denote this initial sum by $\lambda$.
Observe that
Equation \ref{eq:weight} can be rewritten as
\begin{equation}
  (\lambda + t + 1)   \mu_{js_j}^i(k,t+1) = (\lambda + t ) \mu_{js_j}^{i}(k,t)  + \bar{x}_{js_j}(t).
    \label{eq:belief}
\end{equation}
Hence, even though agent $k$ directly updates the weights, its individual state can be characterized by the set of beliefs $\{\bm{\upmu}_j^i(k)\}_{j\in V_i}$. 
In the following, we usually drop the time index $t$ and agent index $k$ in the bracket (depending on the context) for notational convenience.

\section{Belief Dynamics in Population Network Games}
Observe that for an arbitrary agent $k$, its belief $\bm{\upmu}^{i}_{j}(k)$ is in the simplex $\Delta_j = \{ \bm{\upmu}^{i}_{j}(k) \in \mathbb{R}^{|S_j|} | \sum_{s_j \in S_j} \mu_{js_j}^i(k) =1,  
 \mu_{js_j}^i(k)\geq 0, \forall s_j \in S_j \}$.
We assume that the system state is characterized by a Borel probability measure $P$ defined on the state space $\Delta = \prod_{ i \in V} \Delta_i$. 
Given $\bm{\upmu}_i \in \Delta_i$,
we write the marginal probability density function as $p(\bm{\upmu}_i,t)$.
Note that $p(\bm{\upmu}_i,t)$ is the density of agents having the belief $\bm{\upmu}_i$ about population $i$ \emph{throughout} the system.  
Define $\bm{\upmu} = \{\bm{\upmu}_i\}_{i \in V} \in \Delta$.
Since agents maintain separate beliefs about different neighbor populations, the joint probability density function $p(\bm{\upmu},t)$ can be factorized, i.e., $p(\bm{\upmu},t) = \prod_{ i \in V} p(\bm{\upmu}_i, t)$.
We make the following assumption for the initial marginal density functions.
\begin{assumption}
At time $t=0$, for each population $i\in V$, the marginal density function $p(\bm{\upmu}_i, t)$  is continuously differentiable and has zero mass at the boundary of the simplex $\Delta_i$.
\label{as:initialbelief}
\end{assumption}
This assumption is standard and common for a ``nice'' probability distribution.
Under this mild condition, we determine the evolution of the system state $P$ with the following proposition, using the techniques similar to those in \cite{lahkar2013reinforcement,hu2019modelling}.
\begin{proposition}
[\textbf{Population Belief Dynamics}]
The continuous-time dynamics of the marginal density function $p(\bm{\upmu}_i, t)$ for each population $i \in V$  is governed by a partial differential equation
\begin{equation}
    - \frac{ \partial p(\bm{\upmu}_i, t)}{ \partial t} = \nabla \cdot \left(  p(\bm{\upmu}_i ,t) \frac{ \mathbf{\bar{x}}_{i} -\bm{\upmu}_{i}}{\lambda + t + 1}  \right )
    \label{eq:pde}
\end{equation}
 where $\nabla \cdot$ is the divergence operator and  $\bar{\mathbf{x}}_{i}$ is the mean mixed strategy with each $s_i$-th element 
\begin{align}
    \bar{x}_{is_i} & = \int_{\prod_{j\in V_i} \Delta_j} \frac{\exp{(\beta u_{is_i} )}}{\sum_{s_i'\in S_i} \exp{(\beta u_{is_i'} )} } \prod_{j\in V_i} p(\bm{\upmu}_{j}, t) \left(\prod_{j\in V_i} d \bm{\upmu}_{j}\right )
      \label{eq:meanchoice}
\end{align}
where $u_{is_i} = \sum_{j\in V_i} \mathbf{e}_{s_i}^\top \mathbf{A}_{ij} \bm{\upmu}_j$.
\label{theorem:pde}
\end{proposition}
For every marginal density function $p(\bm{\upmu}_i, t)$,  the total mass is always conserved (Corollary 1 of the supplementary); moreover, the mass at the boundary of the simplex $\Delta_i$ always remains zero, indicating that agents' beliefs will never go  to extremes (Corollary 2 of the supplementary).

Generalizing the notion of a system state to a distribution over beliefs allows us to address a very specific question --- the impact of belief heterogeneity on system evolution. 
That said, partial differential equations (Equation \ref{eq:pde}) are notoriously difficult to solve. Here we resort to the evolution of moments based on the evolution of the distribution (Equation \ref{eq:pde}).  
In the following proposition, we show that the characterization of belief heterogeneity is important, as the dynamics of the mean system state (or the mean belief dynamics) is indeed affected by belief heterogeneity.
\begin{proposition}[\textbf{Mean Belief Dynamics}]
\label{prop:meanbelief}
The  dynamics of the mean  belief $\bar{\bm{\upmu}}_i$ about each population $i\in V$ is governed by a system of differential equations such that  for each strategy $s_i$,
\begin{align}
    \frac{d \bar{\mu}_{i s_i}}{d t} & 
       \approx \frac{f_{s_i}(\{\bm{\upmu}_j\}_{j\in V_i} ) - \bar{\mu}_{i s_i} }{\lambda + t + 1} + \frac { \sum_{j\in V_i} \sum_{s_j \in S_j} \frac{\partial^2 f_{s_i}(\{\bm{\upmu}_j\}_{j\in V_i})}{(\partial \mu_{j s_j})^2} \text{Var}(\mu_{j s_j}) }{2(\lambda + t + 1)}.
       \label{eq:meandynamicsTrun}
\end{align}
where  $f_{s_i}(\{\bm{\upmu}\}_{j\in V_i})$ is the logit choice function  (Equation \ref{eq:exploration}) applied to strategy $s_i\in S_i$, and $\text{Var}(\mu_{j s_j})$ is the variance of  belief $\mu_{j s_j}$ in the entire system. 
\end{proposition}
In general, the mean belief dynamics is under the joint effects of the mean, variance, and infinitely many higher moments of the belief distribution.
To allow for more conclusive results, 
we apply the moment closure approximation\footnote{Moment closure is a typical approximation method used to estimate moments of population models \cite{gillespie2009moment,goodman1953population,matis2010achieving}. To use moment closure, a level is chosen past which all cumulants are set to zero. The conventional choice of the level is $2$, i.e., setting the third and higher cumulants to be zero.}  and assume the effects of the third and higher moments to be negligible.

Now, just for a moment, suppose that the system beliefs are homogeneous ----
the beliefs of every individuals are the same.
Hence, the mean belief dynamics are effectively the belief dynamics of individuals. The following proposition follows from Equation \ref{eq:belief}.
\begin{proposition} [\textbf{Belief Dynamics for Homogeneous Populations}]
For a homogeneous system, the dynamics of the belief $\bm{\upmu}_i$ about each population $i\in V$ is governed by a system of  differential equations such that for each strategy $s_i$,
\begin{align}
      \frac{d {\mu}_{is_i} }{d t} =\frac{ {x}_{is_i} - \mu_{is_i}} {\lambda +t+1}
= \frac{f_{s_i}( \{\bm{\upmu}_{j}\}_{j \in V_i} ) - \mu_{is_i} } { \lambda  + t + 1} 
\label{eq:meanbeliefhomo}
\end{align}
where $\mu_{is_i}$ is the same for all agents in each neighbor population $j \in V_i$.
\label{prop:meanbeliefhomo}
\end{proposition}

Intuitively, the mean belief dynamics indicates the trend of beliefs in a system, and the variance of beliefs indicates  belief heterogeneity.
Contrasting Propositions \ref{prop:meanbelief} and \ref{prop:meanbeliefhomo}, it is clear that the variance of belief (belief heterogeneity) plays a role in determining the mean belief dynamics (the trend of beliefs) for heterogeneous systems.
It is then natural to ask: how does the belief heterogeneity evolve over time? How much does the belief heterogeneity affect the trend of beliefs?
Our investigation to these questions reveals an interesting finding  --- the variance of beliefs asymptotically tends to zero.
\begin{theorem}[\textbf{Quadratic Decay of the Variance of Population Beliefs}]
The dynamics of the variance of beliefs $\bm{\upmu}_i$ about each population $i \in V$ is governed by a system of differential equations such that for each strategy $s_i$,
\begin{equation}
    \frac{d \text{Var}(\mu_{is_i})}{dt} =-\frac{ 2 \text{Var}(\mu_{is_i})}{\lambda + t + 1}.
    \label{eq:variance}
\end{equation}
At given time $t$,
 $   \text{Var}(\mu_{is_i}) = \left (\frac{ \lambda + 1 }{\lambda + t + 1}\right )^2 \sigma^2({\mu}_{is_i})$,
where $ \sigma^2(\mu_{is_i})$ is the initial variance.
Thus, the variance $\text{Var}(\mu_{is_i})$  decays to zero quadratically fast with time.
\label{theorem:variance}
\end{theorem}
Such quadratic decay of the variance stands no matter what 2-player subgames agents play and what initial conditions are.
Put differently, the beliefs will eventually homogenize for all population network games.
This fact immediately implies the system state in the limit.
\begin{corollary}
As time $t\to \infty$, the density function $p(\bm{\upmu}_i, t)$ for each population $i \in V$ evolves into a Dirac delta function, and the variance of the choice distributions within each population  $i \in V$  also goes to zero. 
\label{cor:homo}
\end{corollary}
Note that while the choice distributions will homogenize within each population, they are not necessarily the same across different populations.
This is because the strategy choice of each population is in response to its own set of neighbor populations (which are generally different).

\section{Convergence of Smooth Fictitious Play in Population Network Games}
The finding on belief homogenization is non-trivial and also technically important.
One implication is that
the fixed points of systems with initially heterogeneous beliefs are the same as in systems with homogeneous beliefs.
Thus, it follows from the belief dynamics for homogeneous systems (Proposition \ref{prop:meanbeliefhomo}) that the fixed points of systems have the following property. 
\begin{theorem}[\textbf{Fixed Points of System Dynamics}]
For any system that initially have homogeneous or heterogeneous beliefs, the fixed points of the system dynamics is a pair $(\bm{\upmu}^\ast, \mathbf{x}^\ast)$ that satisfy  $\mathbf{x}_i^\ast = \bm{\upmu}_i^\ast$ for each population $i \in V$ and are the solutions of the system of equations
\begin{equation}
x_{is_i}^\ast =  \frac{\exp \left(\beta \sum_{j\in V_i} \mathbf{e}_{s_i}^\top \mathbf{A}_{ij} \mathbf{x}_j^\ast \right)}{\sum_{s_i'\in S_i} \exp \left(\beta \sum_{ j\in V_i} \mathbf{e}_{s_i'}^\top \mathbf{A}_{ij} \mathbf{x}_j^\ast \right) } 
     \label{eq:fixedpoint}
\end{equation}
for every strategy $s_i \in S_i$ and population $i \in V$.
Such fixed points always exist and coincide with the Quantal Response Equilibria (QRE) \cite{mckelvey1995quantal} of the population network game $\Gamma$. 
\label{prop:fixedpoint}
\end{theorem}
Note that the above theorem applies for all population network games. 

We study the convergence of SFP to the QRE under the both cases of network competition and network coordination.
Due to space limits, in the following, we mainly focus on network competition and present only the main result on network coordination. 

\subsection{Network Competition}
Consider a competitive population network game $\Gamma$. 
Note that in competitive network games, the Nash equilibrium payoffs need not to be unique (which is in clear contrast to two-player settings), and it generally allows for infinitely many Nash equilibria.
In the following theorem, focusing on homogeneous systems, we establish the convergence of the belief dynamics to a unique QRE, regardless of the number of Nash equilibria in the underlying game.
\begin{theorem}[\textbf{Convergence in Homogeneous Network Competition}] Given a competitive $\Gamma$, for any system that has  homogeneous beliefs,  the belief dynamics (Equation \ref{eq:meanbeliefhomo}) converges to a unique QRE which is globally asymptotically stable.
\label{theorem:convergeHomo}
\end{theorem}
\begin{hproof}
We proof this theorem by showing that the ``distance'' between $\mathbf{x}_i$ and $\mathbf{\upmu}_i$ is strictly decreasing until  the QRE is reached. In particular, we measure the distance in terms of the perturbed payoff and construct a strict Lyapunov function 
\begin{equation}
     L \coloneqq \sum_{ i \in V} \omega_i \left [ \pi_i\left( \mathbf{x}_i , \{\bm{\upmu}_{j}\}_{j\in V_i} \right ) -   \pi_i\left( \bm{\upmu}_{i} , \{\bm{\upmu}_{j}\}_{j\in V_i} \right ) \right ]
\end{equation}
where $\omega_1 \ldots \omega_n$ are the positive weights given by $\Gamma$, and $\pi_i$ is a perturbed payoff function defined as
$    \pi_i\left( \mathbf{x}_i , \{\bm{\upmu}_{j}\}_{j\in V_i} \right ) \coloneqq \mathbf{x}_i^\top \sum_{j \in V_i} A_{ij} \bm{\upmu}_{j}  -\frac{1}{\beta}\sum_{s_i \in S_i} x_{is_i}\ln (x_{is_i})$.
\end{hproof}
Next, we turn to systems with initially heterogeneous beliefs. 
Leveraging that the variance of beliefs eventually goes to zero, we establish the following lemma. 
\begin{lemma}
\label{le:1}
For a system that initially has heterogeneous beliefs, the mean belief dynamics (Equation \ref{eq:meandynamicsTrun}) is asymptotically autonomous \cite{markus1956asymptotically} with the limit equation 
 $   \frac{d \bm{\upmu}_i }{d t} = \mathbf{x}_{i} - \bm{\upmu}_i ,$
which after time-reparmeterization is equivalent to the belief dynamics for homogeneous systems (Equation \ref{eq:meanbeliefhomo}).
\end{lemma}
For ease of presentation, we follow the convention to denote the solution flows of an asymptotically autonomous system and its limit equation by $\phi$ and $\Theta$, respectively. 
Thieme \cite{thieme1992convergence}  provides the following seminal result that connects the limit behaviors of $\phi$ and $\Theta$.
\begin{lemma}[Thieme  \cite{thieme1992convergence} Theorem 4.2] 
\label{le:2}
Given a metric space $(X,d)$. Assume that the equilibria of  $\Theta$ are isolated compact $\Theta$-invariant subsets of $X$. The $\omega$-$\Theta$-limit set of any pre-compact $\Theta$-orbit contains a $\Theta$-equilibrium. 
The point $(s, x), s \geq t_0, x\in X$, have a pre-compact $\phi$-orbit.
Then the following alternative holds:
1) $\phi(t,s,x) \to e, t \to \infty$, for some $\Theta$-equilibrium e, and 2)
the $\omega$-$\phi$-limit set of $(s,x)$ contains finitely many $\Theta$-equilibria which are chained to each other in a cyclic way.
\end{lemma}
Combining the above results, we prove the convergence for initially heterogeneous systems. 
\begin{theorem} [\textbf{Convergence in Initially Heterogeneous Network Competition}]
Given a competitive $\Gamma$, for any system that initially has heterogeneous beliefs, the mean belief dynamics (Equation \ref{eq:meandynamicsTrun})  converges to a unique QRE.
\label{theorem:convergeHetero}
\end{theorem}
The following corollary immediately follows as the result of belief homogenization.
\begin{corollary}
For any competitive $\Gamma$, under smooth fictitious play, the choice distributions and beliefs of every individual converges to a unique QRE (given in Theorem \ref{prop:fixedpoint}), regardless of belief initialization and the number of Nash equilibria in $\Gamma$.
\label{cor:zerosum}
\end{corollary}

\subsection{Network Coordination}
We delegate most of the results on coordination network games to the supplementary, and summarize only the main result here.
\begin{theorem} [\textbf{Convergence in Network Coordination with Star Structure}]
Given a coordination $\Gamma$ where the network structure consists of a single or disconnected multiple stars, each orbit of the belief dynamics (Equation \ref{eq:meanbeliefhomo}) for homogeneous systems as well as each orbit of the mean belief dynamics (Equation \ref{eq:meandynamicsTrun}) for initially heterogeneous systems converges to the set of QRE.
\label{theorem:convergePotential}
\end{theorem}
Note that this theorem applies to all 2-population coordination games, as network games with or without star structure are essentially the same when there are only two vertices.
We also remark that pure or mixed Nash equilibria in coordination network games are complex; as reported in recent works \cite{cai2011minmax,boodaghians2018smoothed,babichenko2021settling},  finding a pure Nash equilibrium is PLS-complete.
Hence, learning in the general case of network coordination is difficult and generally requires some conditions for theoretical analysis \cite{nagarajan2020chaos,palaiopanos2017multiplicative}.

\section{Experiments: Equilibrium Selection in Population Network Games}
In this section, we complement our theory and present an empirical study of SFP in a two-population coordination (stag hunt) game and a five-population zero-sum (asymmetric matching pennies) game.
Importantly, these two games both have multiple Nash equilibria, which naturally raises the problem of equilibrium selection. 

\begin{figure}[tb!]
\begin{minipage}[b]{0.22\textwidth}
\centering
    \begin{tabular}{|c|c|c|}
    \hline
          & $H$  & $S$  \\ \hline
    $H$ & (1, 1) & (2, 0) \\ \hline
    $S$ & (0, 2) & (4, 4) \\ \hline
\end{tabular}
\captionof{table}{Stag Hunt.}
\end{minipage}
\hfill
\begin{minipage}[b]{0.75\textwidth}
\centering
\includegraphics[width=0.95\textwidth]{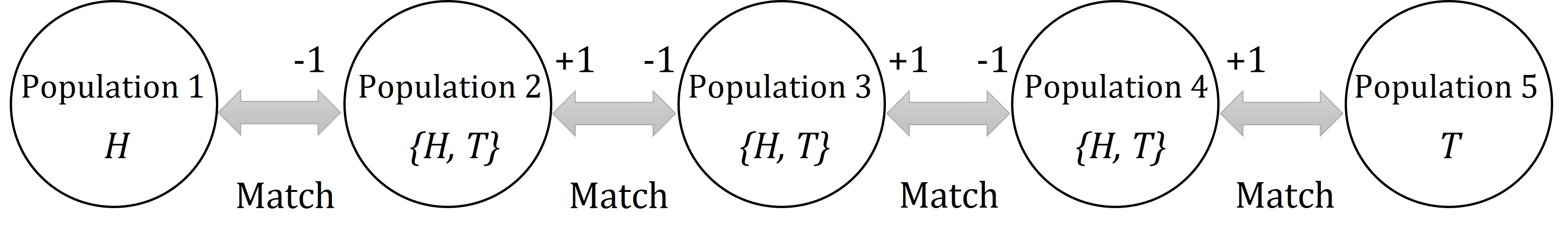}
\caption{Asymmetric Matching Pennies.}
\end{minipage}
\end{figure}

\subsection{Two-Population Stag Hunt Games}
We have shown in Figure \ref{fig:SHgame} (in the introduction) that given the same initial mean belief, changing the variances of initial beliefs can result in  different limit behaviors.
In the following, we systematically study the effect of initial belief heterogeneity by visualizing how it affects the regions of attraction to different equilibria. 

\textbf{Game Description.}
We consider a two-population stag hunt game, where each player in populations $1$ and $2$ has two actions $\{H, S\}$. As shown in the payoff bi-matrices (Table 1), there are two pure strategy Nash equilibria in this game: $(H,H)$ and $(S,S)$. While $(H,H)$ is risk dominant,  $(S,S)$ is indeed more desirable as it is payoff dominant as well as Pareto optimal. 

\textbf{Results.}
In this game, population 1 forms beliefs about population 2 and vice versa. We denote the initial mean beliefs by a pair $(\bar{\mu}_{2H}, \bar{\mu}_{1H})$. 
We numerically solve the mean belief dynamics for a large range of initial mean beliefs, given different variances of initial beliefs. 
In Figure \ref{fig:boa}, for each pair of initial mean beliefs, we color the corresponding data point based on which QRE the system eventually converges to. 
We observe that as the variance of initial beliefs increases (from the left to right panel), a larger range of initial mean beliefs results in the convergence to the QRE that approximates the payoff dominant equilibrium $(S,S)$.
Put differently, a higher degree of initial belief heterogeneity leads to a larger region of attraction to $(S,S)$.
Hence, belief heterogeneity eventually vanishes though, it  provides an approach to equilibrium selection, as it helps select the highly desirable equilibrium.

\begin{figure}[tb!]
    \centering
    \includegraphics[width=\textwidth]{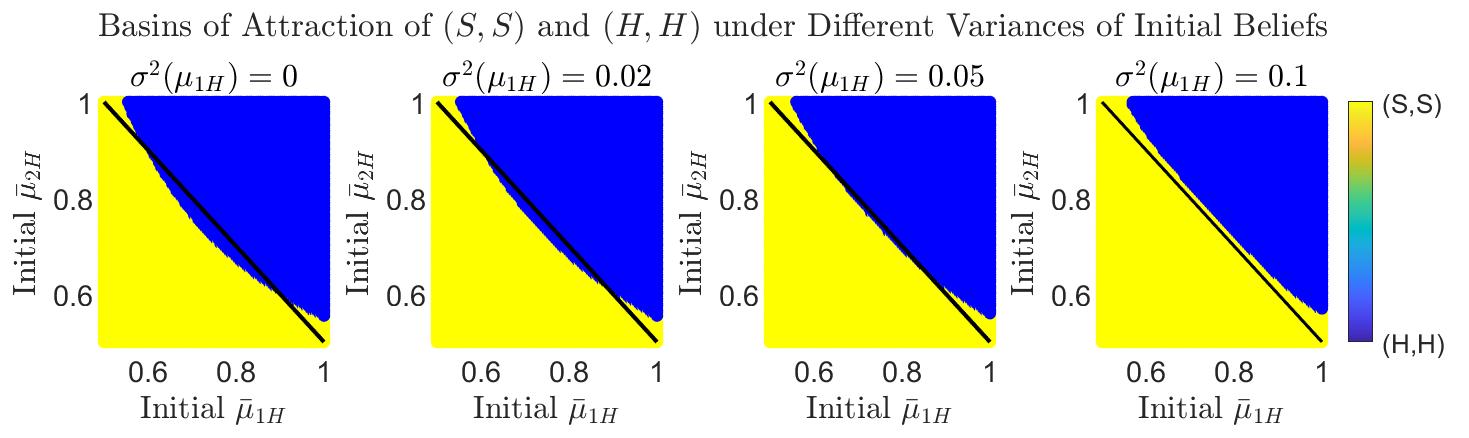}
    \caption{Belief heterogeneity helps select the payoff dominant equilibrium $(S,S)$ (yellow: the equilibrium $(S,S)$, blue: the equilibrium $(H,H)$). As the variance of initial beliefs increases (from the left to right panel), a larger range of initial mean beliefs will approximately reach the equilibrium $(S,S)$ in the limit. For each panel, the initial variances of two populations $\sigma^2(\mu_{1H})$ and $\sigma^2(\mu_{2H})$ are the same.}
    \label{fig:boa}
\end{figure}

\begin{figure}[tb!]
    \centering
    \includegraphics[width=0.95\textwidth]{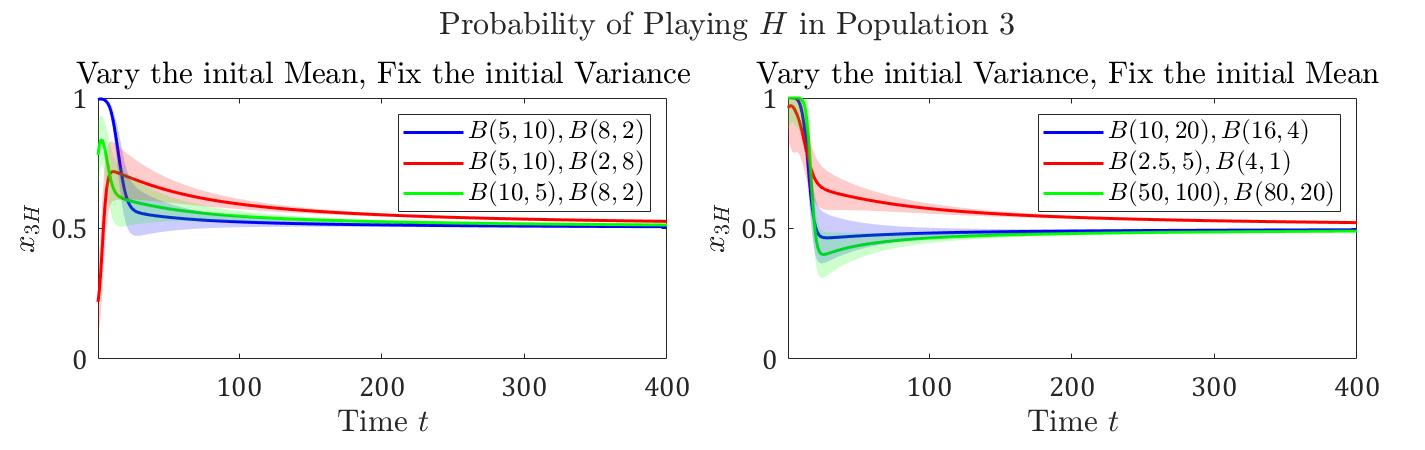}
    \caption{With different belief initialization, SFP selects a unique equilibrium where all agents in population $3$ play strategy $H$ with probability $0.5$. We run 100 simulation runs for each initialization. The thin lines represent the mean mixed strategy (the choice probability of $H$) and the shaded areas represent the variance of the mixed strategies in the population. In the legends, $B$ denotes Beta distribution; the two Beta distributions correspond to the initial beliefs about the neighbor populations $2$ and $4$, respectively.}
    \label{fig:zerosum}
\end{figure}

\subsection{Five-Population Asymmetric Matching Pennies Games}
We have shown in Corollary \ref{cor:zerosum} that SFP converges to a unique QRE even if there are multiple Nash equilibria in a competitive $\Gamma$. 
In the following, we  corroborate this by providing empirical evidence in agent-based simulations with different belief initialization (the details of simulations are summarized in the supplementary). 

\textbf{Game Description.}
Consider a five-population asymmetric matching pennies game \cite{leonardos2021exploration},  where the network structure is a line (depicted in Figure 2). Each agent has two actions $\{H, T\}$.
Agents in populations $1$ and $5$ do not learn; they always play strategies $H$ and $T$, respectively. 
For agents in populations $2$ to $4$, they receive $+1$ if they match the strategy of the opponent in the next population, and receive $-1$ if they mismatch.
On the contrary, they receive $+1$ if they mismatch the strategy of the opponent in the previous population, and receive $-1$ if they match. Hence, this game has infinitely many Nash equilibria of the form: agents in populations $2$ and $4$ play strategy $T$, whereas agents in population $3$ are indifferent between strategies $H$ and $T$. 

\textbf{Results.} 
In this game, agents in each population form two beliefs (one for the previous population and one for the next population). We are mainly interested in the strategies of population $3$, as the Nash equilibria differ in the strategies in population $3$. For validation, we vary population $3$'s beliefs about the neighbor populations $2$ and $4$, and fix population $3$'s beliefs about the other populations.
As shown in Figure \ref{fig:zerosum}, given differential initialization of beliefs, agents in population $3$ converge to the same equilibrium where they all take strategy $H$ with probability $0.5$. 
Therefore, even when the underlying zero-sum game has many Nash equilibria, SFP with different initial belief heterogeneity selects a unique equilibria, addressing the problem of equilibrium selection.

\section{Conclusions}
We study a  heterogeneous beliefs  model of SFP in network games. Representing the system state with a distribution over beliefs, we prove that beliefs eventually become homogeneous in all network games.
We establish the convergence of SFP to Quantal Response Equilibria in general competitive network games as well as coordination network games with star structure. 
We experimentally show that although the initial belief heterogeneity vanishes in the limit, it plays a crucial role in equilibrium selection and helps select highly desirable equilibria.

\section*{Appendix A: Corollaries and Proofs omitted in Section 3}
\subsection*{Proof of Proposition 1}
It follows from Equation 7 in the main paper that the change in $\bm{\upmu}_j^i(k,t)$ between two discrete time steps is 
\begin{equation}
    \bm{\upmu}_j^i(k,t+1)   = \bm{\upmu}_j^i(k,t) + \frac{\bar{\mathbf{x}}_j(t) - \bm{\upmu}_j^i(k,t)} {\lambda + t + 1}.
    \label{eq:muupdate}
\end{equation}
\begin{lemma}
Under Assumption 1 (in the main paper), for an arbitrary agent $k$ in population $i$, its belief $\bm{\upmu}_j^i(k,t)$ about a neighbor population $j$ will never reach the extreme belief (i.e., the boundary of the simplex $\Delta_i$).
\end{lemma}
\begin{proof}
Assumption 1 ensures that $\bar{\mathbf{x}}_j(0)$ is in the interior of the simplex $\Delta_j$. Moreover, the logit choice function (Equation 5 in the main paper) also ensures that  $\bar{\mathbf{x}}_j(t)$ stays in the interior of $\Delta_j$ afterwards for a finite temperature $\beta$. Hence, from Equation \ref{eq:muupdate}, one can see that  $\bm{\upmu}_j^i(k,t)$ for every time step $t$ will stay in the interior of $\Delta_j$.
\end{proof}
In the following, for notation convenience, we sometimes drop the agent index $k$ and the time index $t$ depending on the context. 
Consider a population $i$. We rewrite the change in the beliefs about this population as follows.
\begin{equation}
    \bm{\upmu}_i(t+1)   = \bm{\upmu}_i(t) + \frac{\bar{\mathbf{x}}_i(t) - \bm{\upmu}_i(t)} {\lambda + t + 1}.
\end{equation}

Suppose that the amount of time that passes between two successive time steps is $\delta\in (0,1]$. We rewrite the above equation as 
\begin{equation}
    \bm{\upmu}_i(t+\delta) = \bm{\upmu}_i(t) + \delta \frac{\bar{\mathbf{x}}_i(t) - \bm{\upmu}_i(t)} {\lambda + t + 1}.
        \label{eq:muupdate1}
\end{equation}
Next, we consider a test function $\theta(\bm{\upmu}_i)$. Define
\begin{equation}
\label{eq:quantity}
Y = \frac{ \mathbb{E}[\theta(\bm{\upmu}_i(t+\delta) )] - \mathbb{E}[\theta(\bm{\upmu}_i(t) )]} {\delta}. 
\end{equation}
Applying Taylor series for $\theta( \bm{\upmu}_i(t+\delta) ) $  at $\bm{\upmu}_i(t)$, we obtain
\begin{align}
 \theta( \bm{\upmu}_i(t+\delta) ) & =  \theta(\bm{\upmu}_i(t))   + \frac{\delta}{\lambda + t + 1} \partial_{\bm{\upmu}_i} \theta (\bm{\upmu}_i)    \left [\bar{\mathbf{x}}_i(t) - \bm{\upmu}_i(t) \right] \notag \\ 
 & \quad + \frac{\delta^2}{2(\lambda + t + 1)^2} \left [\bar{\mathbf{x}}_i(t) - \bm{\upmu}_i(t)  \right ]^\top \mathbf{H} \theta (\bm{\upmu}_i) \left [\bar{\mathbf{x}}_i(t) - \bm{\upmu}_i(t)  \right ]  \notag \\
& \quad +  o\left (\left[\delta \frac{\bar{\mathbf{x}}_i(t) - \bm{\upmu}_i(t)} {\lambda + t + 1}\right]^2 \right) 
 \end{align}
where $\mathbf{H}$ denotes the Hessian matrix.
Hence, the expectation  $\mathbb{E}[\theta(\bm{\upmu}_i(t+\delta) )]$ is
\begin{align}
\mathbb{E}[\theta(\bm{\upmu}_i(t+\delta))] &= \mathbb{E}[\theta(\bm{\upmu}_i(t) )] + \frac{\delta}{\lambda + t + 1} 
\mathbb{E}[\partial_{\bm{\upmu}_i} \theta (\bm{\upmu}_i(t))  (\bar{\mathbf{x}}_i(t) - \bm{\upmu}_i(t)) ] \notag \\
& \quad + \frac{\delta^2}{2 (\lambda + t + 1)^2} \mathbb{E}\left[ [\bar{\mathbf{x}}_i(t) - \bm{\upmu}_i(t)   ]^\top \mathbf{H} \theta (\bm{\upmu}_i) \left [\bar{\mathbf{x}}_i(t) - \bm{\upmu}_i(t)  \right ] \right ] \notag \\
& \quad + \frac{\delta^2}{2 (\lambda + t + 1)^2}  \mathbb{E}[o( [\bar{\mathbf{x}}_i(t) - \bm{\upmu}_i(t)]^2) ] 
\end{align}
Moving the term $\mathbb{E}[\theta(\bm{\upmu}_i(t) )] $ to the left hand side and dividing both sides by $\delta$, we recover the quantity $Y$, i.e., 
\begin{align}
Y & = \frac{1}{\lambda + t + 1} 
\mathbb{E}[  \partial_{\bm{\upmu}_i }\theta (\bm{\upmu}_i(t)) (\bar{\mathbf{x}}_i(t) - \bm{\upmu}_i(t)) ] \notag \\
 & \quad + \frac{\delta}{2 (\lambda + t + 1)^2} \mathbb{E}[ [\bar{\mathbf{x}}_i(t) - \bm{\upmu}_i(t)]^\top \mathbf{H}\theta(\bm{\upmu}_i(t)) [\bar{\mathbf{x}}_i(t) - \bm{\upmu}_i(t) ] + o\left((\bar{\mathbf{x}}_i(t) - \bm{\upmu}_i(t))^2 \right) ]
\end{align}
Taking the limit of $Y$ with $\delta \to 0$, the contribution of the second  term on the right hand side vanishes, yielding
\begin{align}
\lim_{\delta \to 0} Y  & = \frac{1}{\lambda + t + 1} 
\mathbb{E}[ \partial_{\bm{\upmu}_i } \theta (\bm{\upmu}_i(t))(\bar{\mathbf{x}}_i(t) - \bm{\upmu}_i(t))] \\
& = \frac{1}{\lambda + t + 1} 
 \int p(\bm{\upmu}_i(t), t) \left [  \partial_{\bm{\upmu}_i} \theta (\bm{\upmu}_i(t)) (\bar{\mathbf{x}}_i(t) - \bm{\upmu}_i(t)) \right ] d \bm{\upmu}_i(t).
\end{align}
Apply integration by parts. We obtain
\begin{align}
\lim_{\delta \to 0} Y   = 0  - \frac{1}{\lambda + t + 1} 
 \int \theta(\bm{\upmu}_i(t)) \nabla \cdot  \left  [ p(\bm{\upmu}_i(t), t) (\bar{\mathbf{x}}_i(t) - \bm{\upmu}_i(t)) \right ] d \bm{\upmu}_i(t) 
\end{align}
where we have leveraged that the probability mass $ p(\bm{\upmu}_i ,t)$ at the boundary $\partial \Delta_i$ remains zero as a result of Lemma 1. 
On the other hand, according to the definition of $Y$,
\begin{equation}
\label{eq:y1}
\begin{aligned}
\lim_{\delta \to 0} Y = \lim_{\delta \to 0} \int \theta(\bm{\upmu}_i(t)) \frac{ p(\bm{\upmu}_i,t+\delta) - p(\bm{\upmu}_i,t) } {\delta} d\bm{\upmu}_i  
= \int \theta(\bm{\upmu}_i(t)) \partial_t  p(\bm{\upmu}_i,t)   d\bm{\upmu}_i.
\end{aligned}
\end{equation}
Therefore, we have the equality
\begin{align}
  \int \theta(\bm{\upmu}_i(t)) \partial_t  p(\bm{\upmu}_i,t)   d\bm{\upmu}_i = -   \frac{1}{\lambda + t + 1} \int \theta(\bm{\upmu}_i(t)) \nabla \cdot  \left  [ p(\bm{\upmu}_i(t), t) (\bar{\mathbf{x}}_i(t) - \bm{\upmu}_i(t)) \right ] d \bm{\upmu}_i(t).
\end{align}
As $\theta$ is a test function, this leads to 
\begin{align}
\partial_t p(\bm{\upmu}_i,t)   
= -   \frac{1}{\lambda + t + 1}  \nabla \cdot  \left  [ p(\bm{\upmu}_i(t), t) (\bar{\mathbf{x}}_i(t) - \bm{\upmu}_i(t)) \right ].
\end{align}
Rearranging the terms, we obtain Equation 8 in the main paper. By the definition of expectation given a probability distribution, it is straightforward to obtain Equation 9 in the main paper. Q.E.D.

\emph{\textbf{Remarks:} The PDEs we derived are akin to the  continuity equation commonly encountered in physics in the study of conserved quantities.The continuity equation describes the transport phenomena (e.g., of mass or energy) in a physical system. This renders a physical interpretation for our PDE model: under SFP, the belief dynamics of a heterogeneous system is  analogously the transport of the agent mass in the simplex $\Delta = \prod_{i \in V}\Delta_i$.}

\subsection*{Corollaries of Proposition 1}
\begin{corollary}
For any population $i \in V$, the system beliefs about this population never go to extremes.
\label{cor:zeromass}
\end{corollary}
\begin{proof}
This is a straightforward result of Lemma 1.
\end{proof}

\begin{corollary}
For any population $i \in V$, the total probability mass $p(\bm{\upmu}_i,t)$ always remains conserved.
\label{cor:conserved}
\end{corollary}
\begin{proof}
Consider the time derivative of the total probability mass 
\begin{align}
    \frac{d}{dt} \int p(\bm{\upmu}_i,t) d\bm{\upmu}_i.
\end{align}
Apply the Leibniz rule to interchange differentiation and integration, 
\begin{align}
    \frac{d}{dt} \int p(\bm{\upmu}_i,t) d\bm{\upmu}_i =  \int   \frac{ \partial p(\bm{\upmu}_i, t)}{ \partial t} d\bm{\upmu}_i. 
\end{align}
Substitute $\frac{ \partial p(\bm{\upmu}_i, t)}{ \partial t}$ with Equation 8 in the main paper,
\begin{align}
   & \frac{d}{dt} \int p(\bm{\upmu}_i,t) d\bm{\upmu}_i  \notag \\
    &  =  - \int   \nabla \cdot \left(  p(\bm{\upmu}_i ,t) \frac{ \mathbf{\bar{x}}_{i} -\bm{\upmu}_{i}}{\lambda + t + 1}  \right )d\bm{\upmu}_i \\
& = - \int  \sum_{s_i\in S_i} \partial_{\mu_{is_i}}\left(  p(\bm{\upmu}_i ,t) \frac{ \bar{x}_{is_i} -\mu_{is_i}}{\lambda + t + 1}  \right )d\bm{\upmu}_i \\
    &=  - \frac{1}{\lambda + t + 1} \left [ \int  \sum_{s_i\in S_i} \partial_{\mu_{is_i}} p(\bm{\upmu}_i ,t) \left(  \bar{x}_{is_i} -\mu_{is_i} \right )d\bm{\upmu}_i  + \int p(\bm{\upmu}_i ,t)  \sum_{s_i\in S_i}  \partial_{\mu_{is_i}} \left(  \bar{x}_{is_i} -\mu_{is_i} \right )  d\bm{\upmu}_i \right ] \label{eq:18}
\end{align}
Apply integration by parts,
\begin{align}
&\int  \sum_{s_i\in S_i} \partial_{\mu_{is_i}} p(\bm{\upmu}_i ,t) \left(  \bar{x}_{is_i} -\mu_{is_i} \right )d\bm{\upmu}_i = 0 -  \int p(\bm{\upmu}_i ,t)  \sum_{s_i\in S_i} \partial_{\mu_{is_i}}\left(  \bar{x}_{is_i} -\mu_{is_i} \right )d\bm{\upmu}_i. 
\end{align}
where we  have leveraged that the probability mass $ p(\bm{\upmu}_i ,t)$ at the boundary $\partial \Delta_i$ remains zero. Hence, the terms within the bracket of Equation \ref{eq:18}  cancel out, and 
\begin{align}
    \frac{d}{dt} \int p(\bm{\upmu}_i,t) d\bm{\upmu}_i = 0.
\end{align}
\end{proof}

\subsection*{Proof of Proposition 2}
\begin{lemma}
The  dynamics of the mean  belief $\bar{\bm{\upmu}}_i$ about each population $i\in V$ is governed by a differential equation
\begin{align}
    \frac{d \bar{\mu}_{is_i} }{d t} =\frac{ \bar{x}_{is_i} -\bar{\mu}_{is_i}} {\lambda +t+1},\qquad \forall s_i \in S_i.
\end{align}
\label{le:meandynamics}
\end{lemma}
\begin{proof}
The time derivative of the mean belief about strategy $s_i$ is 
\begin{align}
  \frac{d  \bar{\mu}_{is_i}  }{d t}  =   \frac{d}{dt} \int \mu_{is_i} p(\bm{\upmu}_i,t) d\bm{\upmu}_i.
\end{align}
We apply the Leibniz rule to interchange differentiation and integration, and then substitute $\frac{ \partial p(\bm{\upmu}_i, t)}{ \partial t}$ with Equation 8 in the main paper.
\begin{align}
    & \frac{d}{dt} \int  \mu_{is_i} p(\bm{\upmu}_i,t) d\bm{\upmu}_i \\
    & =  \int   \mu_{is_i} \frac{ \partial p(\bm{\upmu}_i, t)}{ \partial t} d\bm{\upmu}_i\\ 
    &  = - \int  \mu_{is_i} \nabla \cdot \left(  p(\bm{\upmu}_i ,t) \frac{ \mathbf{\bar{x}}_{i} -\bm{\upmu}_{i}}{\lambda + t + 1}  \right )d\bm{\upmu}_i \\
& = - \int  \mu_{is_i} \sum_{s_i\in S_i} \partial_{\mu_{is_i}}\left(  p(\bm{\upmu}_i ,t) \frac{ \bar{x}_{is_i} -\mu_{is_i}}{\lambda + t + 1}  \right )d\bm{\upmu}_i \\
    &=  \gamma \left [ \int  \mu_{is_i} \sum_{s_i\in S_i} \left(\partial_{\mu_{is_i}} p(\bm{\upmu}_i ,t)\right) \left(  \bar{x}_{is_i} -\mu_{is_i} \right )d\bm{\upmu}_i  +  \int \mu_{is_i} p(\bm{\upmu}_i ,t)  \sum_{s_i\in S_i}  \partial_{\mu_{is_i}} \left(  \bar{x}_{is_i} -\mu_{is_i} \right )  d\bm{\upmu}_i \right ]  \label{eq:28}
\end{align}
where $\gamma \coloneqq - \frac{1}{\lambda + t + 1}$.
Apply integration by parts to the first term in Equation \ref{eq:28}.
\begin{align}
&\int  \mu_{is_i} \sum_{s_i\in S_i} \left ( \partial_{\mu_{is_i}} p(\bm{\upmu}_i ,t) \right ) \left(  \bar{x}_{is_i} -\mu_{is_i} \right )d\bm{\upmu}_i \notag \\
& = -  \int \mu_{is_i}  p(\bm{\upmu}_i ,t)  \left [\sum_{s_i'\in S_i}  \partial_{\mu_{is_i'}}( \bar{x}_{is_i'} -\mu_{is_i'}) \right ] + p(\bm{\upmu}_i ,t)  \partial_{\mu_{is_i}}\left[\mu_{is_i} ( \bar{x}_{is_i} -\mu_{is_i})  \right ] d\bm{\upmu}_i
\end{align}
where we have leveraged that the probability mass at the boundary remains zero.
Hence, it follows from Equation \ref{eq:28} that 
\begin{align}
    & \frac{d}{dt} \int  \mu_{is_i} p(\bm{\upmu}_i,t) d\bm{\upmu}_i \\
  & = - \gamma \int \mu_{is_i}  p(\bm{\upmu}_i ,t)   \sum_{s_i'\in S_i}  \partial_{\mu_{is_i'}}( \bar{x}_{is_i'} -\mu_{is_i'})  d\bm{\upmu}_i  - \gamma \int  p(\bm{\upmu}_i ,t)  \partial_{\mu_{is_i}}\left[\mu_{is_i} ( \bar{x}_{is_i} -\mu_{is_i})  \right ] d\bm{\upmu}_i \notag \\
  & \qquad + \gamma \int \mu_{is_i} p(\bm{\upmu}_i ,t)  \sum_{s_i\in S_i}  \partial_{\mu_{is_i}} \left(  \bar{x}_{is_i} -\mu_{is_i} \right )  d\bm{\upmu}_i \\
  & =  \gamma   \int  p(\bm{\upmu}_i ,t)  \left [ \mu_{is_i}  \partial_{\mu_{is_i}} \left(  \bar{x}_{is_i} -\mu_{is_i} \right ) -  \partial_{\mu_{is_i}}\left[\mu_{is_i} ( \bar{x}_{is_i} -\mu_{is_i})  \right ]  \right ] d\bm{\upmu}_i \\
& = \gamma   \int  p(\bm{\upmu}_i ,t) \mu_{is_i}   d\bm{\upmu}_i  - \int  p(\bm{\upmu}_i ,t) \bar{x}_{is_i}   d\bm{\upmu}_i  \\
    & =  \frac{\bar{x}_{is_i}  - \bar{\mu}_{is_i}  }{\lambda + t + 1} \label{eq:mean}
\end{align}
\end{proof}
We repeat the mean probability $\bar{x}_{is_i}$, which has been given in Equation 9 in the main paper, as follows:
\begin{align}
    \bar{x}_{is_i} & = \int \frac{\exp{(\beta u_{is_i} )}}{\sum_{s_i'\in S_i} \exp{(\beta u_{is_i'} )} } \prod_{j\in V_i} p(\bm{\upmu}_{j}, t) \left(\prod_{j\in V_i} d \bm{\upmu}_{j}\right )
    \label{eq:meanchoice}
\end{align}
where $u_{is_i} = \sum_{j\in V_i} \mathbf{e}_{s_i}^\top \mathbf{A}_{ij} \bm{\upmu}_j$.
Define $\bar{\bm{\upmu}} \coloneqq \{\bar{\bm{\upmu}}_j \}_{j \in V_i}$ and 
\begin{align}
    f_{s_i}(\{\bm{\upmu}_j\}_{j \in V_i}) \coloneqq \frac{ \exp{(\beta \sum_{j\in V_i} \mathbf{e}_{s_i}^\top \mathbf{A}_{ij} \bm{\upmu}_j )} }{\sum_{s_i'\in S_i} \exp{(\beta \sum_{j\in V_i} \mathbf{e}_{s_i'}^\top \mathbf{A}_{ij} \bm{\upmu}_j )} }. 
\end{align}
 Applying the Taylor expansion to approximate this function at the mean belief $\bar{\bm{\upmu}}$, we have
 \begin{align}
      f_{s_i}(\{\bm{\upmu}_j\}_{j \in V_i}) \approx f_{s_i}(\bar{\bm{\upmu}}) + \nabla f_{s_i}(\bar{\bm{\upmu}}) \cdot (\bm{\upmu}- \bar{\bm{\upmu}})  +\frac{1}{2} (\bm{\upmu} - \bar{\bm{\upmu}})^\top \mathbf{H}f_{s_i}(\bar{\bm{\upmu}}) (\bm{\upmu} - \bar{\bm{\upmu}}) + O(||\bm{\upmu} - \bar{\bm{\upmu}} ||^3)
 \end{align}
 where $\mathbf{H}$ denotes the Hessian matrix.
 Hence,  we can rewrite Equation \ref{eq:meanchoice} as
\begin{align}
     \bar{x}_{is_i} &= \int f_{s_i}(\{\bm{\upmu}_j\}_{j \in V_i}) \prod_{j\in V_i} p(\bm{\upmu}_{j}, t) \left(\prod_{j\in V_i} d \bm{\upmu}_{j}\right )\\
     &\approx f_{s_i}(\bar{\bm{\upmu}}) + 
     \int \nabla f_{s_i}(\bar{\bm{\upmu}})\cdot \bm{\upmu}   \prod_{j\in V_i} p(\bm{\upmu}_{j}, t) \left(\prod_{j\in V_i} d \bm{\upmu}_{j}\right )   - \nabla f_{s_i} (\bar{\bm{\upmu}}) \cdot \bar{\bm{\upmu}}    \notag  \\
     & \qquad + \int \frac{1}{2} (\bm{\upmu} - \bar{\bm{\upmu}})^\top \mathbf{H} f_{s_i}(\bar{\bm{\upmu}}) (\bm{\upmu} - \bar{\bm{\upmu}})   \prod_{j\in V_i} p(\bm{\upmu}_{j}, t) \left(\prod_{j\in V_i} d \bm{\upmu}_{j}\right ) \notag \\
     & \qquad + \int O(||\bm{\upmu} - \bar{\bm{\upmu}}||)^3  \prod_{j\in V_i} p(\bm{\upmu}_{j}, t) \left(\prod_{j\in V_i} d \bm{\upmu}_{j}\right ) \label{eq:taylor}
\end{align}
Observe that in Equation \ref{eq:taylor}, the second and the third term can be canceled out. Moreover, for any two neighbor populations $j, k\in V_i$, the beliefs $\bm{\upmu}_j, \bm{\upmu}_k$ about these two populations are separate and independent. Hence, the covariance of these beliefs are zero. We apply the moment closure approximation \cite{matis2010achieving,gillespie2009moment} with the second order and obtain 
\begin{align}
    \bar{x}_{is_i} \approx f_{s_i}(\bar{\bm{\upmu}}) + \frac{1}{2} \sum_{j\in V_i} \sum_{s_j \in S_j} \frac{\partial^2 f_{s_i}(\bar{\bm{\upmu}})}{(\partial \mu_{j s_j})^2} \text{Var}(\mu_{j s_j}).
\end{align}
Hence, substituting $\bar{x}_{is_i}$ in Lemma \ref{le:meandynamics} with the above approximation, we have the mean belief dynamics 
\begin{align}
    \frac{d \bar{\mu}_{i s_i}}{d t} & 
       \approx \frac{f_{s_i}(\bar{\bm{\upmu}} ) - \bar{\mu}_{i s_i} }{\lambda + t + 1} + \frac { \sum_{j\in V_i} \sum_{s_j \in S_j} \frac{\partial^2 f_{s_i}(\bar{\bm{\upmu}} )}{(\partial \mu_{j s_j})^2} \text{Var}(\mu_{j s_j}) }{2(\lambda + t + 1)}.
       \label{eq:meanapprox}
\end{align}
Q.E.D.

\emph{\textbf{Remarks:} the use of the moment closure approximation (considering only the first and the second moments) is for obtaining more conclusive results. Strictly speaking,  the mean belief dynamics also depend on the third and  higher moments. However, we observe in the experiments that these moments in general have little effects on the mean belief dynamics. To be more specific, given the same initial mean beliefs, while the  variance of initial beliefs sometimes can change the limit behaviors of a system, we do not observe similar phenomena for the third and higher moments.   }

\subsection*{Proof of Proposition 3}
Consider a population $i$. 
It follows from Equation 7 in the main paper that the change in the beliefs about this population can be written as follows.
\begin{equation}
    \bm{\upmu}_i(t+1)   = \bm{\upmu}_i(t) + \frac{\mathbf{x}_i(t) - \bm{\upmu}_i(t)} {\lambda + t + 1}.
\end{equation}
Suppose that the amount of time that passes between two successive time steps is $\delta\in (0,1]$. We rewrite the above equation as 
\begin{equation}
    \bm{\upmu}_i(t+\delta) = \bm{\upmu}_i(t) + \delta \frac{\mathbf{x}_i(t) - \bm{\upmu}_i(t)} {\lambda + t + 1}.
\end{equation}
Move the term $\bm{\upmu}_i(t)$ to the right hand side and divide both sides by $\delta$,  
\begin{equation}
\frac{  \bm{\upmu}_i(t+\delta) - \bm{\upmu}_i(t)} {\delta } =   \frac{\mathbf{x}_i(t) - \bm{\upmu}_i(t)} {\lambda + t + 1}.
\end{equation}
Assume that the amount of time $\delta$ between two successive time steps goes to zero. we have
\begin{equation}
  \frac{d \bm{\upmu}_i}{ d t} = \lim_{\delta \to 0} \frac{  \bm{\upmu}_i(t+\delta) - \bm{\upmu}_i(t)} {\delta } =   \frac{\mathbf{x}_i(t) - \bm{\upmu}_i(t)} {\lambda + t + 1}.
\end{equation}
Note that for continuous-time dynamics, we usually drop the time index in the bracket, yielding the belief dynamics (Equation 11) in Proposition 3. Q.E.D.

\subsection*{Proof of Theorem 1}
Without loss of generality, we consider the variance of the belief $\mu_{i s_i}$ about strategy $s_i$ of population $i$. Note that 
\begin{align}
    \text{Var}(\mu_{i s_i}) = \mathbb{E}[(\mu_{i s_i})^2] - (\bar{\mu}_{i s_i})^2.
\end{align}
Hence, we have 
 \begin{equation}
     \frac{d \text{Var}(\mu_{i s_i})}{dt} = \frac{d \mathbb{E}[(\mu_{i s_i})^2] }{d t} - 2 \bar{\mu}_{i s_i} \frac{d \bar{\mu}_{i s_i}}{dt}.
    \label{eq:variance}
 \end{equation}
Consider the first term on the right hand side.
 We apply the Leibniz rule to interchange differentiation and integration, and then substitute $\frac{ \partial p(\bm{\upmu}_i, t)}{ \partial t}$ with Equation 8 in the main paper.
\begin{align}
    &\frac{d \mathbb{E}[(\mu_{i s_i})^2] }{d t} \notag \\
    &= \int  (\mu_{is_i})^2  \frac{\partial p(\bm{\upmu}_i,t)} {\partial t} d\bm{\upmu}_i \\
    &=  - \int  (\mu_{is_i})^2    \nabla \cdot \left(  p(\bm{\upmu}_i ,t) \frac{ \mathbf{\bar{x}}_{i} -\bm{\upmu}_{i}}{\lambda + t + 1}  \right )d\bm{\upmu}_i \\
    &= - \int   (\mu_{is_i})^2  \sum_{s_i\in S_i} \partial_{\mu_{is_i}}\left(  p(\bm{\upmu}_i ,t) \frac{ \bar{x}_{is_i} -\mu_{is_i}}{\lambda + t + 1}  \right )d\bm{\upmu}_i \\
    &=  \gamma \int   (\mu_{is_i})^2 \sum_{s_i\in S_i} \partial_{\mu_{is_i}} p(\bm{\upmu}_i ,t) \left(  \bar{x}_{is_i} -\mu_{is_i} \right )d\bm{\upmu}_i   + \gamma   \int  (\mu_{is_i})^2  p(\bm{\upmu}_i ,t)  \sum_{s_i\in S_i}  \partial_{\mu_{is_i}} \left(  \bar{x}_{is_i} -\mu_{is_i} \right )  d\bm{\upmu}_i  \label{eq:46}
\end{align}
where $\gamma \coloneqq - \frac{1}{\lambda + t + 1}$.
Applying integration by parts to the first term in Equation \ref{eq:46} yields 
\begin{align}
&\int  (\mu_{is_i})^2\sum_{s_i\in S_i} \partial_{\mu_{is_i}} p(\bm{\upmu}_i ,t) \left(  \bar{x}_{is_i} -\mu_{is_i} \right ) d\bm{\upmu}_i \notag \\
& = -  \int (\mu_{is_i})^2  p(\bm{\upmu}_i ,t)  \left [\sum_{s_i'\in S_i}  \partial_{\mu_{is_i'}}( \bar{x}_{is_i'} -\mu_{is_i'}) \right ] + p(\bm{\upmu}_i ,t)  \partial_{\mu_{is_i}}\left[(\mu_{is_i})^2 ( \bar{x}_{is_i} -\mu_{is_i})  \right ] d\bm{\upmu}_i
\end{align}
where  we have leveraged that the probability mass at the boundary remains zero. 
Combining the above two equations, we obtain 
\begin{align}
    &\frac{d \mathbb{E}[(\mu_{i s_i})^2] }{d t} \notag \\
    & = - \gamma    \int (\mu_{is_i})^2  p(\bm{\upmu}_i ,t)  \left [\sum_{s_i'\in S_i}  \partial_{\mu_{is_i'}}( \bar{x}_{is_i'} -\mu_{is_i'}) \right ] + p(\bm{\upmu}_i ,t)  \partial_{\mu_{is_i}}\left[(\mu_{is_i})^2 ( \bar{x}_{is_i} -\mu_{is_i})  \right ] d\bm{\upmu}_i \notag  \\
    & \qquad  + \gamma   \int  (\mu_{is_i})^2  p(\bm{\upmu}_i ,t)  \sum_{s_i\in S_i}  \partial_{\mu_{is_i}} \left(  \bar{x}_{is_i} -\mu_{is_i} \right )  d\bm{\upmu}_i \\
    &= \gamma \int \left [ - p(\bm{\upmu}_i ,t)  \partial_{\mu_{is_i}}\left[(\mu_{is_i})^2 ( \bar{x}_{is_i} -\mu_{is_i})  \right ]  \right ] +  (\mu_{is_i})^2  p(\bm{\upmu}_i ,t)   \partial_{\mu_{is_i}} \left(  \bar{x}_{is_i} -\mu_{is_i} \right )  d\bm{\upmu}_i  \\
    & = \gamma \int 2 (\mu_{is_i})^2  p(\bm{\upmu}_i ,t)   d\bm{\upmu}_i  - \gamma \int 2\bar{x}_{is_i} \mu_{is_i}  p(\bm{\upmu}_i ,t)   d\bm{\upmu}_i \\
    & =  - \frac{  2\mathbb{E}[(\mu_{is_i})^2] -2 \bar{x}_{is_i} \bar{\mu}_{is_i} }{\lambda + t + 1} \label{eq:55}.
\end{align}
Next, we consider the second term in Equation \ref{eq:variance}. By Lemma \ref{le:meandynamics}, we have 
\begin{align}
    2 \bar{\mu}_{i s_i} \frac{d \bar{\mu}_{i s_i}}{dt} = \frac {2 \bar{\mu}_{i s_i} (\bar{x}_{is_i} - \bar{\mu}_{is_i}) } {\lambda + t + 1 } \label{eq:56}.
\end{align}
Combining Equations \ref{eq:55} and \ref{eq:56}, the dynamics of the variance is
\begin{align}
     \frac{d \text{Var}(\mu_{i s_i})}{dt} & =  - \frac{  2\mathbb{E}[(\mu_{is_i})^2] -2 \bar{x}_{is_i} \bar{\mu}_{is_i} }{\lambda + t + 1} - \frac {2 \bar{\mu}_{i s_i} (\bar{x}_{is_i} - \bar{\mu}_{is_i}) } {\lambda + t + 1 } \\
     & = \frac{ 2 (\bar{\mu}_{is_i})^2 - 2\mathbb{E}[(\mu_{is_i})^2]}{\lambda + t + 1 } \\
      & = - \frac{ 2 \text{Var}(\mu_{i s_i})}{\lambda + t + 1 }.
\end{align}
Q.E.D.

\emph{\textcolor{black}{\textbf{Remarks:}
 We believe that the rationale behind such a phenomenon is twofold: 1) agents apply smooth fictitious play, and 2) agents respond to the mean strategy play of other populations rather than the strategy play of some fixed agents. Regarding the former, we notice that under a similar setting, population homogenization may not occur if agents apply other learning methods, e.g., Q-learning and Cross learning. Regarding the latter, imagine that agents adjust their beliefs in response to the strategies of some fixed agents. For example, consider two populations; one contains agents A and C, and the other one contains agents B and D. Suppose that agents A and B form a fixed pair such that they adjust their beliefs only in response to each other; the same applies to agents C and D. Belief homogenization may not happen. }}

\section*{Appendix B: Proofs omitted in Section 4.1}
\subsection*{Proof of Theorem 2}
Belief homogenization implies that the fixed points of systems with initially heterogeneous beliefs are the same as in systems with homogeneous beliefs. Thus, we focus on homogeneous systems to analyze the fixed points.
It is straightforward to see that 
\begin{align}
     \frac{d \bm{\upmu}_i}{ d t}  =   \frac{\mathbf{x}_i - \bm{\upmu}_i} {\lambda + t + 1} = 0 \implies \mathbf{x}_i = \bm{\upmu}_i.
\end{align}
 Denote the fixed points of the system dynamics, which satisfies the above equation, by $(\mathbf{x}_i^\ast,  \bm{\upmu}_i^\ast)$ for each population $i$.
 By the logit choice function (Equation 5 in the main paper), we have
 \begin{align}
     x_{is_i}^\ast =  \frac{\exp{(\beta u_{is_i} )}}{\sum_{s_i'\in S_i} \exp{(\beta u_{is_i'} )} } = \frac{ \exp{(\beta \sum_{j\in V_i} \mathbf{e}_{s_i}^\top \mathbf{A}_{ij} \bm{\upmu}_j^\ast )} }{\sum_{s_i'\in S_i} \exp{(\beta \sum_{j\in V_i} \mathbf{e}_{s_i'}^\top \mathbf{A}_{ij} \bm{\upmu}_j^\ast )} }.
 \end{align}
Leveraging that $\mathbf{x}_i^\ast = \bm{\upmu}_i^\ast, \forall i \in V$ at the fixed points, we can replace $\bm{\upmu}_j^\ast $ with $\mathbf{x}_j^\ast$. Q.E.D.

\subsection*{Proof of Theorem 3}
Consider a population $i$. 
The set of neighbor populations is $V_i$, the set of beliefs about the neighbor populations is  $\{\bm{\upmu}_{j}\}_{j\in V_i}$, and the choice distribution is $\mathbf{x}_i$. 
Given a population network game $\Gamma$, the expected payoff  is given by $\mathbf{x}_i^\top \sum_{(i,j) \in E} A_{ij} \bm{\upmu}_{j}$. 
Define a perturbed payoff function 
\begin{align}
     \pi_i\left( \mathbf{x}_i , \{\bm{\upmu}_{j}\}_{j\in V_i} \right ) \coloneqq \mathbf{x}_i^\top \sum_{j \in V_i} A_{ij} \bm{\upmu}_{j}  +  v( \mathbf{x}_i)
\end{align}
where $v( \mathbf{x}_i) =  -\frac{1}{\beta}\sum_{s_i \in S_i} x_{is_i}\ln (x_{is_i})$. Under this form of $v(\mathbf{x}_i)$, the maximization of $\pi_i$  yields the choice distribution $\mathbf{x}_i$ from the logit choice function \cite{fudenberg1998theory}. 
Based on this, we establish the following lemma.
\begin{lemma}
\label{le:perturbpayoff}
For a choice distribution $\mathbf{x}_i$ of SFP in a population network game, 
\begin{equation}
    \partial_{\mathbf{x}_i} \pi_i\left( \mathbf{x}_i , \{\bm{\upmu}_{j}\}_{j\in V_i} \right ) = \mathbf{0} \quad \text{and} \quad \sum_{j \in V_i} \left( A_{ij} \bm{\upmu}_{j} \right)^\top = -\partial_{\mathbf{x}_i} v( \mathbf{x}_i).
\end{equation}
\begin{proof}
This lemma immediately follows from the fact that the maximization of $\pi_i$ will yield the choice distribution $\mathbf{x}_i$ from the logit choice function \cite{fudenberg1998theory}.
\end{proof}
\end{lemma}
The belief dynamics of a homogeneous populations can be simplified after time-reparameterization.
\begin{lemma}
Given $\tau  = \ln\frac{\lambda + t + 1 }{\lambda + 1}$,
the belief dynamics of homogeneous systems (given in Equation 11 in the main paper) is equivalent to
\begin{align}
          \frac{d \bm{\upmu}_i }{d \tau} = \mathbf{x}_{i} - \bm{\upmu}_i.
\end{align}
\end{lemma}
\begin{proof}
From $\tau  = \ln\frac{\lambda + t + 1 }{\lambda + 1}$, we have
\begin{align}
    t = (\lambda + 1)(\exp{(\tau)} - 1).
\end{align}
By the chain rule, for each dimension $s_i$, 
\begin{align}
    \frac{d \mu_{i s_i}}{ d \tau} & =      \frac{d \mu_{i s_i}}{ d t} \frac{dt}{d \tau}   \\
& = \frac{x_{i s_i} - \mu_{i s_i}}{\lambda + t + 1} \frac{d \left((\lambda + 1)(\exp{(\tau)} - 1) \right )}{d \tau}\\
& = \frac{x_{i s_i} - \mu_{i s_i}}{\lambda + (\lambda + 1)(\exp{(\tau)} - 1) + 1} (\lambda + 1) \exp{(\tau)}\\
& = x_{i s_i} - \mu_{i s_i}.
\end{align}
\end{proof}
Next, we define the Lyapunov function $L$ as 
\begin{equation}
    L \coloneqq \sum_{ i \in V} \omega_i L_i \quad \text{s.t.} \quad  L_i \coloneqq  \pi_i\left( \mathbf{x}_i , \{\bm{\upmu}_{j}\}_{j\in V_i} \right ) -   \pi_i\left( \bm{\upmu}_{i} , \{\bm{\upmu}_{j}\}_{j\in V_i} \right ).
\end{equation}
where $\{\omega_i\}_{i\in V}$ is the set of positive weights defined in the weighted zero-sum $\Gamma$.
The function $L$ is non-negative because for every $i\in V$, $\mathbf{x}_i$ maximizes the function $\pi_i$. When for every $i \in V$, $\mathbf{x}_i =  \bm{\upmu}_{i}$, the function $L$ reaches the minimum value $0$.

Rewrite $L$ as
\begin{equation}
    L = \sum_{i \in V} \left [ \omega_i  \pi_i\left( \mathbf{x}_i , \{\bm{\upmu}_{j}\}_{j\in V_i} \right ) - \omega_i \bm{\upmu}_i^\top \sum_{j\in V_i} A_{ij} \bm{\upmu}_j - \omega_i v(\bm{\upmu}_i) \right ].
\end{equation}
We observe that $\pi_i\left( \mathbf{x}_i , \{\bm{\upmu}_{j}\}_{j\in V_i} \right ) $ is convex in $\bm{\upmu}_{j}, j\in V_i$  by Danskin's theorem, and $- v(\bm{\upmu}_i)$ is strictly convex in $\bm{\upmu}_i$.
Moreover,  by the weighted zero-sum property given in Equation 2 in the main  paper, we have
\begin{align}
    \sum_{i \in V} \left(\omega_i \bm{\upmu}_i^\top \sum_{j \in V_i} A_{ij} \bm{\upmu}_j \right) = 0
\end{align}
since $\mu_i \in \Delta_i, \mu_j \in \Delta_j$ for every $i,j \in V.$
Therefore, the function $L$ is a strictly convex function and attains its minimum value $0$ at a unique point  $\mathbf{x}_i =  \bm{\upmu}_{i}$, $\forall i \in V.$ 

Consider the function $L_i$. Its time derivative is
\begin{equation}
\begin{aligned}
    \dot{L}_i& =  
\partial_{\mathbf{x}_i} \pi_i\left( \mathbf{x}_i , \{\bm{\upmu}_{j}\}_{j\in V_i} \right ) \dot{\mathbf{x}}_i
+ \sum_{j\in V_i}  \left [ \partial_{\bm{\upmu}_j} \pi_i\left( \mathbf{x}_i , \{\bm{\upmu}_{j}\}_{j\in V_i} \right ) \dot{\bm{\upmu}}_j \right ] \\
&  \quad - \partial_{\bm{\upmu}_{i}} \pi_i\left( \bm{\upmu}_{i} , \{\bm{\upmu}_{j}\}_{j\in V_i} \right ) \dot{\bm{\upmu}}_i
- \sum_{j\in V_i}  \left [ \partial_{\bm{\upmu}_j} \pi_i\left( \bm{\upmu}_{i}, \{\bm{\upmu}_{j}\}_{j\in V_i} \right ) \dot{\bm{\upmu}}_j \right ]. 
\label{eq:timelyapunov}
\end{aligned}
\end{equation}
Note that the partial derivative $\partial_{\mathbf{x}_i} \pi_i$ equals $\mathbf{0}$ by Lemma \ref{le:perturbpayoff}. Thus, we can rewrite this as
\begin{align}
    \dot{L}_i & = \partial_{\bm{\upmu}_{i}} \pi_i\left( \bm{\upmu}_{i} , \{\bm{\upmu}_{j}\}_{j\in V_i} \right ) \dot{\bm{\upmu}}_i
    + \sum_{ j\in V_i} \left [ \partial_{\bm{\upmu}_j} \pi_i\left( \mathbf{x}_i , \{\bm{\upmu}_{j}\}_{j\in V_i} \right ) - \partial_{\bm{\upmu}_j} \pi_i\left( \bm{\upmu}_{i}, \{\bm{\upmu}_{j}\}_{j\in V_i} \right )  \right ] \dot{\bm{\upmu}}_j  \\
    & = - \left [\sum_{j\in V_i} \left (A_{ij} \bm{\upmu}_j \right )^\top + \partial_{\mathbf{\upmu}_i}v(\bm{\upmu}_i) \right ] (\mathbf{x}_i - \bm{\upmu}_i)
    + \sum_{ j\in V_i} \left (\mathbf{x}_i^\top A_{ij}  - \bm{\upmu}_i^\top  A_{ij} \right ) ( \mathbf{x}_j -\bm{\upmu}_j ) \label{eq:1}\\
     & =   \left [  \partial_{\mathbf{x}_i} v(\mathbf{x}_i)  - \partial_{\mathbf{\upmu}_i} v(\bm{\upmu}_i) \right ] (\mathbf{x}_i - \bm{\upmu}_i) + \sum_{j\in V_i} \left (\mathbf{x}_i^\top A_{ij} \mathbf{x}_j - \bm{\upmu}_i^\top A_{ij}  \mathbf{x}_j - \mathbf{x}_i^\top A_{ij} \bm{\upmu}_j + \bm{\upmu}_i^\top A_{ij}  \bm{\upmu}_j \right)\label{eq:2}.
\end{align}
where from Equation \ref{eq:1} to \ref{eq:2}, we  apply Lemma \ref{le:perturbpayoff} to substitute $\sum_{j \in V_i} \left (A_{ij} \bm{\upmu}_j \right )^\top$ with $-\partial_{\mathbf{x}_i} v(\mathbf{x}_i)$.
Hence, summing over all the populations, the time derivative of $L$ is
\begin{align}
    \dot{L} &= \sum_{ i \in V} \omega_i \left [  \partial_{\mathbf{x}_i}v(\mathbf{x}_i)  - \partial_{\mathbf{\upmu}_i} v(\bm{\upmu}_i) \right ] (\mathbf{x}_i - \bm{\upmu}_i) \notag \\
    & \quad + \sum_{i \in V} \sum_{j\in V_i} \omega_i \left (\mathbf{x}_i^\top A_{ij} \mathbf{x}_j - \bm{\upmu}_i^\top A_{ij}  \mathbf{x}_j - \mathbf{x}_i^\top A_{ij} \bm{\upmu}_j + \bm{\upmu}_i^\top A_{ij}  \bm{\upmu}_j \right). 
\end{align}
The summation in the second line is equivalent to
\begin{align}
  &  \sum_{(i,j)\in E} 
     (\omega_i \mathbf{x}_i^\top A_{ij} \mathbf{x}_j +  \omega_j \mathbf{x}_j^\top A_{ji} \mathbf{x}_i) - (\omega_i \bm{\upmu}_i^\top A_{ij}  \mathbf{x}_j +
\omega_j \mathbf{x}_j^\top A_{ji} \bm{\upmu}_i ) \\
& \ \qquad - (\omega_i \mathbf{x}_i^\top A_{ij} \bm{\upmu}_j +  \omega_j  \bm{\upmu}_j^\top A_{ji}  \mathbf{x}_i) + (\omega_i  \bm{\upmu}_i^\top A_{ij}  \bm{\upmu}_j + \omega_j\bm{\upmu}_j^\top A_{ji} \bm{\upmu}_i).
\end{align}
By the weighted zero-sum property given in Equation 2 in the main paper, this summation equals $0$, yielding 
\begin{equation}
     \dot{L} = \sum_{ i \in V} \omega_i \left [  \partial_{\mathbf{x}_i} v(\mathbf{x}_i)  - \partial_{\mathbf{\upmu}_i} v(\bm{\upmu}_i) \right ] (\mathbf{x}_i - \bm{\upmu}_i). 
\end{equation}
Note that the function $v$ is strictly concave such that its second derivative is negative definite.
By this property, $\dot{L} \leq 0$ with equality only if $\mathbf{x}_i =\bm{\upmu}_i , \forall i \in V$, which corresponds to the  QRE. Therefore, $L$ is a strict Lyapunov function, and the global asymptotic stability of the QRE follows. Q.E.D.

\emph{\textbf{Remarks:} Intuitively, the Lyapunov function defined above measures the distance between the QRE and a given set of beliefs. The idea of measuring the distance in terms of entropy-regularized payoffs is inspired from the seminal work  \cite{hofbauer2005learning}. However, different from the network games considered in this paper, Hofbauer and Hopkins \cite{hofbauer2005learning} consider SFP in two-player games. To our knowledge, so far there has been no systematic study on SFP in network games.   } 

\subsection*{Proof of Theorem 4}
The proof of Theorem 4 leverages the seminal results of the  asymptotically autonomous dynamical system \cite{markus1956asymptotically,thieme1992convergence,thieme1994asymptotically} which  conventionally is defined as follows.
\begin{definition}
A nonautonomous  system of differential equations in $R^n$
\begin{equation}
    x' = f(t, x)
   \label{eq:nonauto}
\end{equation}
is said to be \textit{asymptotically autonomous} with \textit{limit equation}
\begin{equation}
    y' = g(y),
    \label{eq:limit}
\end{equation}
if $f(t, x) \to g(x), t \to \infty,$ 
where the convergence is uniform on each compact subset of $R^n$.
Conventionally, the solution flow of Eq. \ref{eq:nonauto} is called the asymptotically autonomous semiflow (denoted by $\phi$) and the  solution flow of Eq. \ref{eq:limit}  is called the limit semiflow (denoted by $\Theta$).
\end{definition}
Based on this definition, we establish Lemma 1 in the main paper, which is repeated as follows.
\begin{lemma}
\label{le:limit}
For a system that initially has heterogeneous beliefs, the mean belief dynamics is asymptotically autonomous \cite{markus1956asymptotically} with the limit equation 
\begin{align}
       \frac{d \bm{\upmu}_i }{d t} = \mathbf{x}_{i} - \bm{\upmu}_i \label{eq:limitBelief}
\end{align}
which after time-reparameterization is equivalent to the belief dynamics for homogeneous systems.
\end{lemma}
\begin{proof}
 We first time-reparameterize the mean belief dynamics of heterogeneous systems. Assume $\tau  = \ln\frac{\lambda + t + 1 }{\lambda + 1}$.
By the chain rule and Equation \ref{eq:meanapprox}, for each dimension $s_i$,
\begin{align}
    \frac{d \bar{\mu}_{i s_i}}{d \tau} &= \frac{d \bar{\mu}_{i s_i}}{d t} \frac{dt}{d\tau} \\
    & = \left [ \frac{f_{s_i}(\bar{\bm{\upmu}} ) - \bar{\mu}_{i s_i} }{\lambda + t + 1} + \frac { \sum_{j\in V_i} \sum_{s_j \in S_j} \frac{\partial^2 f_{s_i}(\bar{\bm{\upmu}})}{(\partial \mu_{j s_j})^2} \text{Var}(\mu_{j s_j}) }{2(\lambda + t + 1)}  \right ]  \frac{d \left((\lambda + 1)(\exp{(\tau)} - 1) \right )}{d \tau}\\
    &=  \frac{f_{s_i}(\bar{\bm{\upmu}} ) - \bar{\mu}_{i s_i} + \frac{1}{2} \sum_{j\in V_i} \sum_{s_j \in S_j} \frac{ \partial^2 f_{s_i}(\bar{\bm{\upmu}})}{(\partial \mu_{j s_j})^2} \left(\frac{\lambda+1}{\lambda+t+1}\right)^2  \sigma^2(\mu_{js_j}) }{\lambda + (\lambda + 1)(\exp{(\tau)} - 1) + 1}  \left(\lambda + 1 \right )\exp{(\tau)}\\
    & = f_{s_i}(\bar{\bm{\upmu}} ) - \bar{\mu}_{i s_i} + \frac{1}{2} \sum_{j\in V_i} \sum_{s_j \in S_j} 
    \frac{\partial^2 f_{s_i}(\bar{\bm{\upmu}})} {(\partial \mu_{j s_j})^2} \sigma^2(\mu_{j s_j}) \exp{(-2\tau)} . \label{eq:88}
\end{align}
Observe that $\exp{(-2\tau)}$ decays to zero exponentially fast and that both $\sigma^2(\mu_{js_j}) $ and $\frac{\partial^2 f_{s_i}(\bar{\bm{\upmu}})}{(\partial \mu_{j s_j})^2}$ 
are bounded for every $\bm{\upmu}$ in the simplex $\prod_{j\in V_i} \Delta_j$.
Hence, Equation \ref{eq:88} converges locally and uniformly to the following equation: 
\begin{align}
    \frac{d \bar{\mu}_{i s_i}}{d \tau} = f_{s_i}(\bar{\bm{\upmu}} ) - \bar{\mu}_{i s_i}
    \label{eq:89}
.\end{align}
Note that $x_{i s_i}=f_{s_i}(\bar{\bm{\upmu}} )$ for homogeneous systems, and the above equation is  algebraically equivalent to Equation \ref{eq:limitBelief}. Hence, by Definition 1, Equation \ref{eq:88} is asymptotically autonomous with the limit equation being Equation \ref{eq:limitBelief}.
\end{proof}
By the above lemma, we can formally connect the limit behaviors of initially heterogeneous systems and those of homogeneous systems. Recall that Theorem 3 in the main paper states that under SFP, 
there is a unique rest point (QRE) for the belief dynamics in a weighted zero-sum network game $\Gamma$; this excludes the case where there are finitely many equilibria that are chained to each other. Hence, combining Lemma 2 in the main paper, we prove that the mean belief dynamics of initially heterogeneous systems converges to a unique QRE. Q.E.D.

\section*{Appendix C: Results and Proofs omitted in Section 4.2}
For the case of network coordination, we consider   networks that consist of a star or  disconnected multiple stars due to technical reasons.
In Figure 1, we present examples of the considered network structure with different numbers of nodes (populations).
\begin{figure}[tb!]
    \centering
    \includegraphics[width=\textwidth]{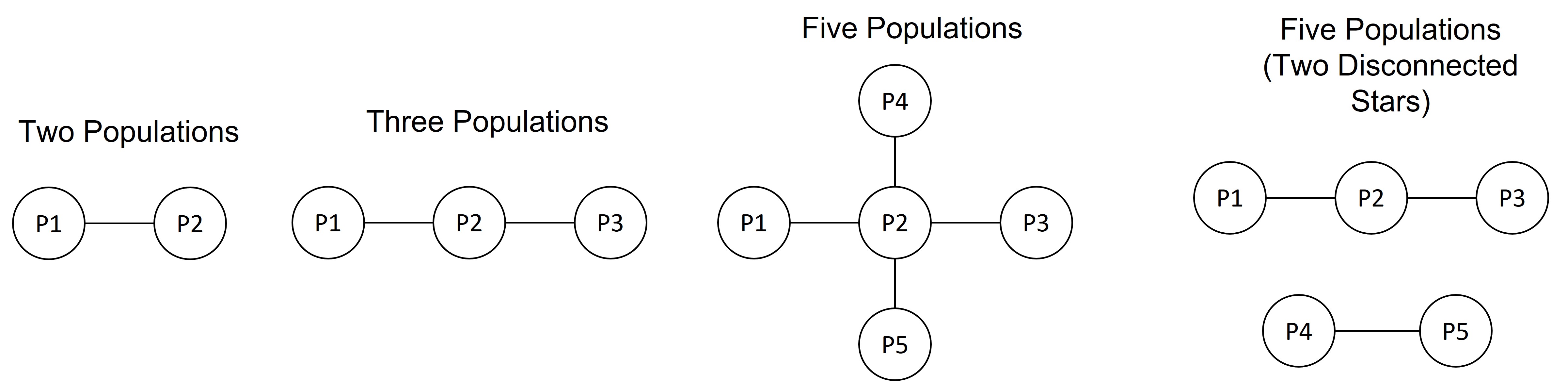}
    \caption{Population network games where the underlying network consists of star structure.}
    \label{fig:star}
\end{figure}

In the following theorem, focusing on homogeneous systems, we establish the convergence of the belief dynamics to the set of QRE.
\begin{theorem} [\textbf{Convergence in Homogeneous Network Coordination with Star Structure}]
Given a coordination $\Gamma$ where the network structure consists of a single or disconnected multiple stars, each orbit of the belief dynamics for homogeneous systems converges to the set of QRE.
\label{theorem:convergehomo}
\end{theorem}
\begin{proof}
Consider a root population $j$  of a star structure. 
Its set of leaf (neighbor) populations is $V_j$, the set of beliefs about the leaf populations is  $\{\bm{\upmu}_{i}\}_{i\in V_j}$, and the choice distribution is $\mathbf{x}_j$. 
Given the game $\Gamma$, the expected payoff  is $\mathbf{x}_j^\top \sum_{i\in V_j} A_{ji} \bm{\upmu}_{i}$. 
Define a perturbed payoff function 
\begin{align}
     \pi_j\left( \mathbf{x}_j , \{\bm{\upmu}_{i}\}_{i\in V_j} \right ) \coloneqq \mathbf{x}_j^\top \sum_{i \in V_j} A_{ji} \bm{\upmu}_{i}  +  v( \mathbf{x}_j)
\end{align}
where $v( \mathbf{x}_j) =  -\frac{1}{\beta}\sum_{s_j \in S_j} x_{j s_j}\ln (x_{j s_j})$.  Under this form of $v(\mathbf{x}_j)$, the maximization of $\pi_j$  yields the choice distribution $\mathbf{x}_j$ from the logit choice function \cite{fudenberg1998theory}. 

Consider a leaf population $i$ of the root population $j$. It has only one neighbor population, which is population $j$. Thus, given the game $\Gamma$, the expected payoff is  $\mathbf{x}_i^\top  A_{ij} \bm{\upmu}_{j}$. Define a perturbed payoff function
\begin{align}
     \pi_i\left( \mathbf{x}_i , \bm{\upmu}_{j} \right ) \coloneqq \mathbf{x}_i^\top  A_{ij} \bm{\upmu}_{j} +  v( \mathbf{x}_i)
\end{align}
where $v( \mathbf{x}_i) =  -\frac{1}{\beta}\sum_{s_i \in S_i} x_{is_i}\ln (x_{is_i})$. Similarly, the maximization of $\pi_i$ yields the choice distribution $\mathbf{x}_i$ from the logit choice function \cite{fudenberg1998theory}. 
Based on this, we establish the following lemma.

\begin{lemma}
For choice distributions of SFP in a population network game with start structure, 
\begin{align}
    &\partial_{\mathbf{x}_j} \pi_j\left( \mathbf{x}_j , \{\bm{\upmu}_{i}\}_{i\in V_j} \right ) = \mathbf{0} \quad \text{and} \quad \sum_{i \in V_j} \left( A_{ji} \bm{\upmu}_{i} \right)^\top = -\partial_{\mathbf{x}_j}v( \mathbf{x}_j) \quad &\text{if $j$ is a root population,} \\
    &\partial_{\mathbf{x}_i} \pi_i\left( \mathbf{x}_i , \bm{\upmu}_{j} \right ) = \mathbf{0} \quad  \text{and} \quad \left( A_{ij} \bm{\upmu}_{j} \right)^\top = -\partial_{\mathbf{x}_i}v( \mathbf{x}_i) \quad &\text{if $i$ is a leaf population.}
\end{align}
\begin{proof}
This lemma immediately follows from the fact that the maximization of $\pi_j$ and $\pi_i$ , respectively, yield the choice distributions $\mathbf{x}_j$ and $\mathbf{x}_i$ from the logit choice function \cite{fudenberg1998theory}.
\end{proof}
\label{le:starpayoff}
\end{lemma}
For readability, we repeat the belief dynamics of a homogeneous population after time-reparameterization, which has been proved in Lemma 4 in Appendix B, as follows:
\begin{align}
          \frac{d \bm{\upmu}_i }{d \tau} = \mathbf{x}_{i} - \bm{\upmu}_i.
\end{align}
 
Let $\mathcal{R}\subset V$ be the set of all root populations. We define 
\begin{align}
L \coloneqq \sum_{j \in \mathcal{R}}  L_j \quad \text{s.t.} \quad
    L_j
    & \coloneqq \bm{\upmu}_j^\top  \sum_{ i\in V_j} A_{ji} \bm{\upmu}_i + v(\bm{\upmu}_j) + \sum_{ i \in V_j} v(\bm{\upmu}_i).
\end{align}
Consider the function $L_j$. Its time derivative  $\dot{L}_j$ is
\begin{align}
& \dot{L}_j =\left [ \partial_{\bm{\upmu}_j} (\bm{\upmu}_j^\top   \sum_{ i\in V_j} A_{ji} \bm{\upmu}_i) \dot{\bm{\upmu}}_j + \sum_{i \in V_j} \partial_{\bm{\upmu}_i} (\bm{\upmu}_j^\top  \sum_{ i\in V_j} A_{ji} \bm{\upmu}_i) \dot{\bm{\upmu}}_i \right ] + \partial_{\mathbf{\upmu}_j}v(\bm{\upmu}_j)\dot{\bm{\upmu}}_j +  \sum_{ i\in V_j} \partial_{\mathbf{\upmu}_i}v(\bm{\upmu}_i)\dot{\bm{\upmu}}_i \\
 &  =  \sum_{ i\in V_j} ( A_{ji} \bm{\upmu}_i)^\top (\mathbf{x}_{j} - \bm{\upmu}_j)  +  \left [ \sum_{i\in V_j} \bm{\upmu}_j^\top A_{ji} (\mathbf{x}_i - \bm{\upmu}_i) \right ] + \partial_{\mathbf{\upmu}_j}v(\bm{\upmu}_j) (\mathbf{x}_{j} - \bm{\upmu}_j) +   \sum_{ i\in V_j} \partial_{\mathbf{\upmu}_i} v(\bm{\upmu}_i) (\mathbf{x}_i -\bm{\upmu}_i). 
\end{align}
Since $\Gamma$ is a coordination game, we have $\left (A_{ij} \bm{\upmu}_j\right)^\top = \bm{\upmu}_j^\top A_{ij}^\top = \bm{\upmu}_j^\top A_{ji}$.
Hence, applying Lemma \ref{le:starpayoff}, we can substitute
$\sum_{ i\in V_j}  (A_{ji} \bm{\upmu}_i)^\top$ with $-v'(\mathbf{x}_{j})$, and $\bm{\upmu}_j^\top A_{ji}$ with   $-v'(\mathbf{x}_{i})$, yielding 
\begin{align}
    & \dot{L}_j = -\partial_{\mathbf{x}_j} v(\mathbf{x}_{j}) (\mathbf{x}_{j} - \bm{\upmu}_j)+   \left [ \sum_{ i\in V_j} (-\partial_{\mathbf{x}_i}v(\mathbf{x}_i ))(\mathbf{x}_i - \bm{\upmu}_i) \right] + \partial_{\mathbf{\upmu}_j} v(\bm{\upmu}_j) (\mathbf{x}_{j} - \bm{\upmu}_j) +   \sum_{ i\in V_j}\partial_{\mathbf{\upmu}_i} v(\bm{\upmu}_i) (\mathbf{x}_i -\bm{\upmu}_i) \\
    & =(\partial_{\mathbf{\upmu}_j}v(\bm{\upmu}_j) - \partial_{\mathbf{x}_j}v(\mathbf{x}_{j})) (\mathbf{x}_{j} - \bm{\upmu}_j) + \sum_{i\in V_j} (\partial_{\mathbf{\upmu}_i}v(\bm{\upmu}_i) - \partial_{\mathbf{x}_i}v(\mathbf{x}_i)) (\mathbf{x}_i -\bm{\upmu}_i)
\end{align}
Note that the function $v$ is strictly concave such that its second derivative is negative definite.
By this property,  $\dot{L}_{j} \geq 0$ with equality only if $\mathbf{x}_i = \bm{\upmu}_i, \forall i \in V_j$ and $\mathbf{x}_{j} = \bm{\upmu}_j$.
Thus, the time derivative of the function $L$, i.e.,  $\dot{L} = \sum_{j\in \mathcal{R}} \dot{L}_j \geq 0$ with equality only if $\mathbf{x}_i = \bm{\upmu}_i, \forall i \in V_j, \mathbf{x}_{j} = \bm{\upmu}_j, \forall j \in \mathcal{R}$.
\end{proof}

We generalize the convergence result to initially heterogeneous systems in the following theorem. 
\begin{theorem} [\textbf{Convergence in Initially Heterogeneous Network Coordination with Star Structure}]
Given a coordination $\Gamma$ where the network structure consists of a single or disconnected multiple stars, each orbit of the mean belief dynamics for initially heterogeneous systems converges to the set of QRE.
\label{theorem:convergehetero}
\end{theorem}
\begin{proof}
The proof technique is similar to that for initially heterogeneous competitive network games. By Lemma 1 in the main paper, we show that the mean belief dynamics of initially heterogeneous systems is asymptotically autonomous with the belief dynamics of homogeneous systems. Therefore, it follows from  Lemma 2 in the main paper that the convergence result for homogeneous systems can be carried over to the initially heterogeneous systems. 
\end{proof}

\emph{\textbf{Remarks}: The convergence of SFP in coordination games and potential games has been established under the 2-player settings \cite{hofbauer2005learning} as well as some n-player settings \cite{hofbauer2002global,swenson2019smooth}.
Our work differs from the previous works in two aspects.
First, our work allows for heterogeneous beliefs.
Moreover, we consider that agents maintain separate beliefs about other agents, while in the previous works agents do not distinguish between other agents. Thus, even when the system beliefs are homogeneous, our setting is still different from (and more complicated) than the previous settings. }

\section*{Appendix D: Omitted Experimental Details}

\paragraph{Numerical Method for the PDE model.} PDEs are notoriously difficult to solve, and only limited types of PDEs allow analytic solutions. Hence, similar to previous research \cite{hu2019modelling}, we resort to numerical method for PDEs; in particular, we  consider the finite difference method \cite{smith1985}.  
\paragraph{Agent-based Simulations.} The presented simulation results are averaged over 100 independent simulation runs to smooth out the randomness. For each simulation run, there are $1,000$ agents in each population. For each agent, the initial beliefs are sampled from the given initial probability distribution. 

\paragraph{Detailed Experimental Setups for Figure 1.}
In the case of small initial variance, the initial beliefs $\mu_{1H}$ and $\mu_{2H}$ are distributed according to the distribution $\text{Beta}(280, 120)$. On the contrary, in the case of large initial variance, the initial beliefs $\mu_{1H}$ and $\mu_{2H}$  are distributed according to the distribution $\text{Beta}(14, 6)$.
Thus, initially, the mean beliefs in these two cases are both $\bar{\mu}_{1H} = \bar{\mu}_{2H} = 0.7$ and $\bar{\mu}_{1S} = \bar{\mu}_{2S} = 0.3$.
In both cases, the initial sum of weights $\lambda = 10$ and the temperature $\beta = 10$.

\paragraph{Detailed Experimental Setups for Figure 3.}
We visualize the regions of attraction of different equilibria in stag hunt games by 
numerically solving the mean belief dynamics (Equation 10 in the main paper). The initial variances have  been given in the title of each panel. In all cases, the initial sum of weights $\lambda = 0$ and the temperature $\beta = 5$.  

\paragraph{Detailed Experimental Setups for Figure 4.}
We let the initial beliefs about populations 1, 3 and 5 remain unchanged across different cases, and vary the initial beliefs about populations 2 and 4.
The initial beliefs about populations 1, 3 and 5, denoted by $\mu_{1H}$, $\mu_{3H}$ and $\mu_{5H}$, are distributed according to the distributions $\text{Beta}(20,10)$, $\text{Beta}(6,4)$, and $\text{Beta}(10,5)$, respectively.
The initial beliefs about populations 2 and 4 have been given in the legends of Figure 4. In all cases, the initial sum of weights $\lambda = 10$ and the temperature $\beta = 10$.
Note that $\mu_{iT} = 1-\mu_{iH}$ for all populations $i = 1,2,3,4,5.$

\paragraph{Source Code and Computing Resource.} We have attached the source code for reproducing our main experiments. The Matlab script \emph{finitedifference.m}  numerically solves our PDE model presented in Proposition 1 in the main paper.
The Matlab script \emph{regionofattraction.m}  visualizes the region of attraction of different equilibria in stag hunt games, which are presented in Figure 3. The Python scripts \emph{simulation(staghunt).py} and \emph{simulation(matchingpennies).py} correspond to the agent-based simulations in two-population stag hunt games and five-population asymmetric matching pennies games, respectively.
We use a laptop (CPU: AMD Ryzen 7
5800H) to run all the experiments.

\bibliographystyle{plain}
\bibliography{ref}
\end{document}